\newcommand{\GDB}{G}
\newcommand{\aP}{\mathcal{P}}

\newcommand{\C}{\mathcal{C}}
\newcommand{\dpp}{Def}
\newcommand{\U}{\mathcal{U}}
\newcommand{\A}{\mathcal{A}}
\newcommand{\K}{\mathcal{K}}
\newcommand{\D}{\mathcal{D}}
\newcommand{\Tt}{\mathcal{T}}
\newcommand{\N}{\mathcal{N}}
\newcommand{\E}{\mathcal{E}}

\newcommand{\V}{\mathcal{V}}
\newcommand{\Trans}{\Upsilon}
\newcommand{\R}{\mathcal{R}}
\newcommand{\Le}{\mathcal{L}}
\newcommand{\M}{\mathcal{M}}

\newcommand{\pathto}{\Longrightarrow}
\newcommand{\edge}{\longrightarrow}

\newcommand{\deriv}{\rightsquigarrow}

\documentclass{lmcs} 

\keywords{Graph Databases, Property Graphs, Path Properties}

\usepackage{hyperref}
\theoremstyle{plain} 

\usepackage{color}
\usepackage[arrow,matrix,tips]{xy}
\usepackage[arrow,matrix,tips]{xypic}
\usepackage{tikz}
\usepackage{verbatim} 
\usepackage{pslatex} 
\usepackage{caption}
\usepackage[mathscr]{euscript}
\usepackage{longtable}

\graphicspath{{./images/}}

\usetikzlibrary{external,quotes,positioning,calc,arrows,decorations.markings,positioning,fit,shapes.misc,shapes.multipart,shadows,arrows.meta} 

\tikzstyle{bag} = [align=center] 

\tikzset{nodetype/.style={shape= rectangle split, rectangle split parts=2, fill=white, draw=red, text width=2.75cm,rounded corners, font=\small}}
\tikzset{subnodetype/.style={shape= rectangle split, rectangle split parts=2, fill=white, draw=red, text width=2.75cm,rounded corners, font=\small}}
\tikzset{edgetype/.style={shape= rectangle, fill=gray!5, draw=gray}}
\tikzset{top/.style={shape= rectangle split, rectangle split parts=2, fill=gray!5, draw=gray!50, text width=3.75cm,rounded corners, font=\small}}
\tikzset{bottom/.style={shape= rectangle split, rectangle split parts=2, fill=gray!3, draw=gray!40, text width=3.75cm,rounded corners, font=\small}}

\begin{document}

\title[Properties for Paths in Graph Databases]{Properties for Paths in Graph Databases}
\thanks{This work is partially supported by  MCIN/ AEI /10.13039/501100011033 under grant PID2020-112581GB-C21.}	

\author[F.~Orejas]{Fernando Orejas\lmcsorcid{0000-0002-3023-4006}}[a]
\author[B.~Pino]{Elvira Pino\lmcsorcid{0000-0003-3376-5096}}[a]
\author[R.~Angles]{Renzo Angles\lmcsorcid{0000−0002−6740−9711}}[b]  
\author[E.~Pasarella]{Edelmira Pasarella\lmcsorcid{0000-0001-8315-4977}}[a]
\author[N.~Mylonakis]{Nikos Mylonakis\lmcsorcid{0000-0002-2535-8573}}[a]

\address{Universitat Politècnica de Catalunya, Barcelona, Spain}	
\email{orejas@cs.upc.edu, elvira@cs.upc.edu, edelmira@cs.upc.edu, nicos@cs.upc.edu}  

\address{Universidad de Talca, Chile}	
\email{rangles@utalca.cl}  





\begin{abstract}
 \noindent This paper presents a formalism for defining properties of paths in graph databases, which can be used to restrict the number of solutions to navigational queries. In particular, our formalism allows us to define quantitative properties such as length or accumulated cost, which can be used as query filters. Furthermore, it enables the identification and removal of paths that may be considered ill-formed. 
 
 The new formalism is defined in terms of an operational semantics for the query language that incorporates these new constructs, demonstrating its soundness and completeness by proving its compatibility with a simple logical semantics. We also analyze its expressive power, showing that path properties are more expressive than register automata. Finally, after discussing some complexity issues related to this new approach, we present an empirical analysis carried out using our prototype implementation of the graph database that serves as a running example throughout the paper. The results show that queries using path properties as filters outperform standard queries that do not use them.
\end{abstract}

\maketitle

\section{Introduction}
 
A graph database (GDB, for short) is essentially a database in which data is stored in a graph and whose queries are typically \emph{navigational}, meaning that solving them requires traversing parts of the graph. For example, we may search for a node that is connected to other nodes via certain paths or edges, which in turn may be connected to additional nodes, and so on. Additionally, queries may request values \emph{stored} in (or, according to standard terminology, \emph{properties} of) these nodes and edges.
In this sense, graph databases are particularly useful in scenarios where the topology of the data is as important as the data itself.

In GDB queries, there are many cases where users may need to ask not only about the existence of a path between two nodes but also about the data stored in the nodes and edges involved in such a path \cite{garcia2024}. In this context, it is natural to consider that paths, like nodes and edges, should also have properties so that queries can refer to them. For example, in a graph database of travel connections (flights, trains, etc.), having properties such as length and cost for paths composed of a sequence of flights would allow us to query for routes from Barcelona to Los Angeles with fewer than three connections and costing less than 1000 euros. 

Moreover, not all paths representing a travel connection can be considered well-formed. For example, if a flight is expected to arrive later than the departure time of the following flight, the corresponding path should be considered ill-formed, and therefore discarded. 

There is, however, an essential difference between path properties and node/edge properties: while the latter are stored in the database, path properties must be computed. Hence, we must specify how to compute them. In our proposal, this is done using constraints and \emph{path unfolding}. Roughly speaking, to define a property of paths, we inductively specify the value of that property for a single-edge path using one constraint and for multi-edge paths using another.

 For example, suppose we want to define a property called $length$ for paths,  indicating the number of edges in a given path. If $p$ is a variable denoting a path, then for a path consisting of a single edge, we specify the value of its $length$ using the constraint $p.length == 1$. If $p$ is of the form $y \cdot p'$, where $y$ denotes the first edge and $p'$ the remainder of the path, we define its length with the constraint $p.length == 1 + p'.length$. Notice that this resembles a recursive definition in many programming languages. However, the semantics is quite different,  as we will explain in Sect. \ref{ss: compdef}.
 
 Our approach is formalized by defining an operational semantics for the query language, which abstractly describes how queries are evaluated. To ensure the adequacy of this semantics, we also define a very simple logical semantics, which describes when a query answer is considered correct. We then show that the operational semantics is sound and complete with respect to the logical semantics. 
 
 Furthermore, we explore the expressive power of our approach by comparing it with other approaches \cite{LibkinMV16,BarceloFL15}  based on register automata \cite{KaminskiF94}.

  We believe that our approach is both simple and relatively easy to use. The simplicity stems from decoupling the description of the form of a path (using regular expressions) from the specification of path properties (using sets of constraints). This is similar to the simplicity of constraint logic programming (CLP for short) \cite{JaffarM94}, decoupling the logic programming part of CLP from value computation (that is, SLD resolution is decoupled from constraint solving). For similar reasons, our approach is easy to use. A graph database user already knows how to define queries to obtain the desired information. Then, adding path properties to limit the search is not difficult, since their definition often resemble recursive definitions in programming. In our view, a unified formalism combining both aspects would be more complex to define, use, and possibly implement. 
In this sense, our approach offers a potential answer to the question posed in \cite{LibkinGPC} regarding how to extend GQL with arithmetic conditions on path contents. However, adapting our method to GQL remains as future work.

We have implemented our approach using Ciao Prolog \cite{Ciao12}, and built a small graph database to test its practicality. The results are excellent: some simple queries that cannot be solved in hours without path properties are answered in a few seconds or less when path properties are used to filter unwanted results. Although we have not conducted a detailed complexity analysis, our empirical analysis demonstrates a significant improvement in the execution performance of navigational queries that include path properties.

The paper is structured as follows. Section \ref{running-example}  presents an example of a database that we use throughout the paper as a running example.  In Section \ref{sect:prelim}, we introduce the framework of our work. To simplify the presentation, this framework is fairly basic: databases are property graphs, and queries are conjunctive regular path queries (CRPQs) \cite{CruzMW87,consens1990graphlog} with logical conditions to filter out results that do not satisfy them. Section \ref{def:pproperties} is dedicated to presenting how to define path properties, while Sections \ref{sec:queries} and  \ref{subsec:OPSem}  present the logical and operational semantics of queries in the context of a set of path property definitions. Section \ref{sec:sound-compl} then shows the equivalence of these two semantic definitions.  This is followed by Section \ref{sec:expr}, where we study the expressiveness of our approach,  showing that it is more expressive than register automata. Then, in Section \ref{sec:comp}, after briefly discussing some complexity issues of our approach, we describe an empirical analysis carried out using our prototype implementation. In particular, we show that executing queries with filters based on path properties outperforms executing the same queries without these filters. Finally, Section \ref{sec:conc} presents our conclusion and outlines future work. An appendix provides detailed information on the experiments conducted with our implementation.


\begin{figure}[t]
    \centering
    \resizebox{\textwidth}{!}{%
  \tikzsetnextfilename{instance}%
  \begin{tikzpicture}[->,>=stealth',shorten >=1pt,auto,node distance=2cm,
                        thick, every text node part/.style={align=center}, scale=\textwidth]

%
\node[nodetype] (LAX) 
    {\textbf{$\mathbf{n_1}$:Airport}
     \nodepart{two}
       code = LAX\\
       loc = LosAngeles
     };
\node[nodetype] (LHR) [right= 5cm of LAX] 
    {\textbf{$\mathbf{n_2}$: Airport}
     \nodepart{two}
       code = LHR\\
       loc = London
     };
\node[nodetype] (CDG) [right= 3.5cm of  LHR]
    {\textbf{$\mathbf{n_3}$: Airport}
     \nodepart{two}
       code = CDG\\
       loc = Paris
     };
\node[nodetype] (BCN) [below = 3.5cm of LHR]
    {\textbf{$\mathbf{n_5}$: Airport, TrainSt}
     \nodepart{two}
       code = BCN\\
       loc = Barcelona\\
     };
\node[nodetype] (JFK) [below = 3.5cm of LAX]
    {\textbf{$\mathbf{n_4}$: Airport}
     \nodepart{two}
       code = JFK\\
       loc = NewYork
     };
\node[nodetype] (Sants) [below = 2.cm of BCN]
    {\textbf{$\mathbf{n_6}$: TrainSt
    }
     \nodepart{two}
       code = Sants\\
       loc = Barcelona
     };
%

%
\path[->] (Sants) edge[]     node [edgetype, text width=3cm,  midway, xshift=21ex]  
    {\textbf{$\mathbf{e_{1}}$: byTrain}\\
            \begin{flushleft}
            trainline = Airport \\
                price = 5€ 
            \end{flushleft}
    } (BCN); 
\path[->] (LHR) edge[]     node [edgetype, text width=4.4cm, midway, yshift=18ex]  
    {\textbf{$\mathbf{e_2}$: Flight}
            \begin{flushleft}
                id = AA123 \\
       	    airline = AmericanAirlines \\
       	    price = 700€ \\
      	    dep = 21:00\\
       	    arr = 8:00
            \end{flushleft} 
    } (LAX); 
\path[->] (CDG) edge[]     node [edgetype, text width=3.25cm, midway, midway, yshift=18ex]  
\path[->] (LHR) edge[]     node [edgetype, text width=4cm, xshift=-25ex, midway ]  
    {\textbf{$\mathbf{e_4}$: Flight}\\
            \begin{flushleft}
                id = BA322 \\
                airline = British Airways\\
                price = 150€ \\
                dep = 8:00\\
                arr = 11:00
            \end{flushleft}
    } (BCN); 
\path[->] (BCN) edge[]     node [edgetype, text width=2.75cm, midway, sloped, yshift=-19ex]  
    {\textbf{$\mathbf{e_5}$: Flight}\\
            \begin{flushleft}
                id = IB451 \\
                airline = Iberia\\
                price = 150€ \\
                dep = 10:00\\
                arr = 11:30
            \end{flushleft}
    }(CDG); 
\path[->] (BCN) edge[]     node [edgetype, text width=2.5cm,  midway, yshift=-1ex]  
    {\textbf{$\mathbf{e_6}$: Flight}\\
            \begin{flushleft}
                 id = IB651 \\
      	        airline = Iberia\\
       	        price = 650€ \\
       	        dep = 9:00\\
       	        arr = 15:00
       	\end{flushleft}
    } (JFK); 
\path[->] (JFK.120) edge[]     node  [edgetype, text width=2.2cm, midway, xshift=-1ex ]  
    {\textbf{$\mathbf{e_7}$: Flight}\\
            \begin{flushleft}
                id = AA004 \\
                price = 300€ \\
                dep = 17:00\\
                arr = 21:00
            \end{flushleft}
    }  (LAX.-120); 
\path[->] (JFK.60) edge[]     node [edgetype, text width=2.2cm,  midway, xshift=15ex]  
    {\textbf{$\mathbf{e_8}$: Flight}\\
            \begin{flushleft}
                id = AA003 \\
                price = 400€ \\
                dep = 11:00\\
                arr = 15:00
            \end{flushleft}
    } (LAX.-60); 
\end{tikzpicture}

    }
    \caption{The example graph database }
    \label{fig:PPG-Instance}
\label{fig:running-example}
\end{figure}
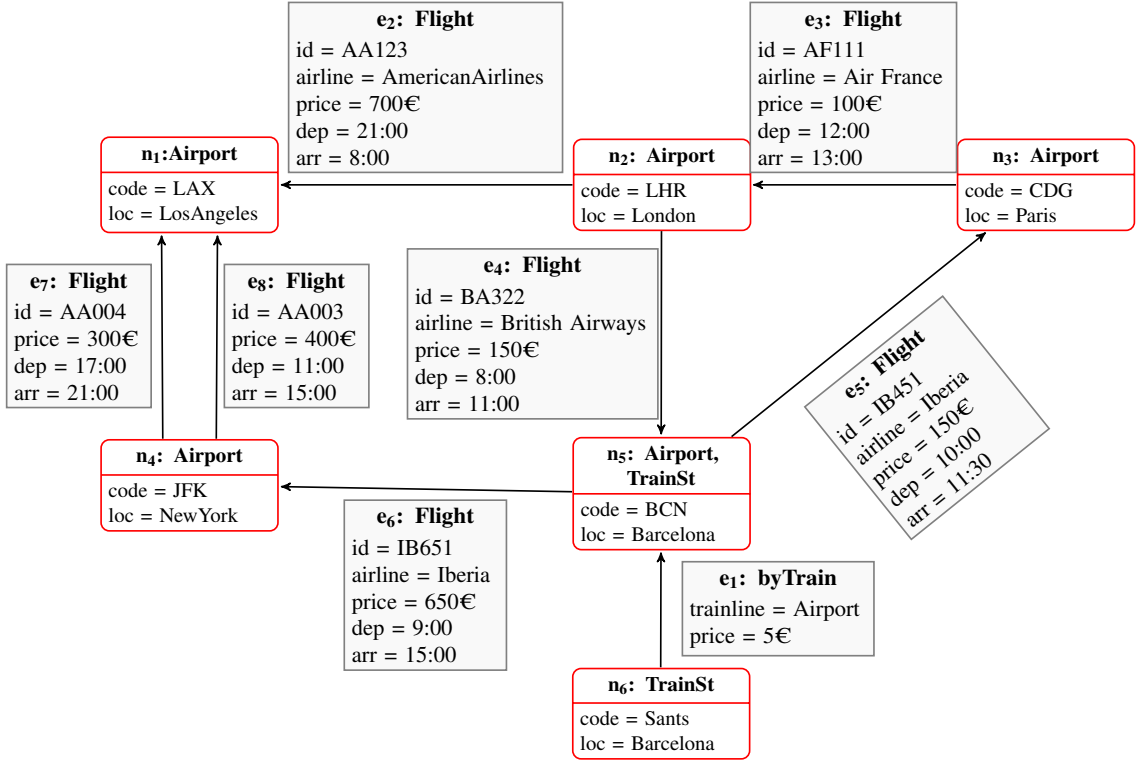

\section{A running example}\label{running-example}

A GDB is a labeled multigraph with named holders to store values, called \emph{properties}, on nodes and edges. Labels, which can also be assigned to nodes and edges, play the role of types. More precisely, a node or edge may have zero or more labels. In the former case, we consider that the element (node or edge) to be untyped or, equivalently, to have any type. In the latter case, we consider that an element with $n$ labels has $n$ types.

For example, Fig. \ref{fig:PPG-Instance} depicts a fragment of the database that we will use as a running example throughout the paper. This database describes the connections between a number of train stations and airports, by means of flights and trains.  In particular, in the figure, $n_1, \dots, n_4$  are nodes labeled by (of type) $Airport$, $n_5$ is a node labeled by $Airport$ and $TrainSt$, and $n_6$ is a node whose label is $TrainSt$. Edges $e_2, \dots, e_8$ are labeled by $Flight$, while $e_{1}$ is an edge labeled by $byTrain$. Additionally, the boxes associated with the nodes and edges include their properties and the values of these properties. For example, $n_4$ has two properties, $code$ and $loc$ (for location). The value stored in $code$ is the string ``BCN''  and the value in $loc$ is ``Barcelona". Similarly, edge $e_7$ has four properties ``id", with the value ``AA004"; $price$, with the value 300; $dep$ (for departure time), with the value 17:00; and $arr$ (for  arrival time), with the value 21:00

A \emph{query} to a GDB is a sequence of \emph{clauses} that describe a fragment of the graph that we are searching for. In particular, each clause consists of a \emph{pattern} specifying an element (a node, an edge, or a path) of the database, together with a \emph{ filter} that establishes certain requirements on the components of the pattern. A simple example of a query is as follows:

\begin{center} \( 
\begin{array}{ll}
(x_1:TrainSt) \ \Box \ x_1.loc ==\hbox{``Barcelona''} \\
(x_1:TrainSt\stackrel{y:byTrain}{\edge} x_2:Airport) \  \Box \ x_2.loc ==\hbox{``Barcelona''}\\ 
(x_2:\!Airport\stackrel{p:Flight^+}{\pathto}x_3:Airport) \  \Box \  x_3.loc ==\hbox{``Los Angeles''}, p.cost < 1000, p.length \le 3\\ 
%
\end{array}
\) 
\end{center}

The query asks for train stations located in Barcelona (first clause), then for direct train connections from these stations to airports in Barcelona (second clause), and, finally, for flight connections from these airports to airports in Los Angeles (third clause), satisfying that the number of connections must be at most 3 and the overall cost must be cheaper than 1000 euros. 

More in detail, the first clause includes a \emph{node pattern} specifying that we are looking for nodes $x_1$ of type $TrainSt$. In our database in Fig. \ref{fig:PPG-Instance} there are only two such nodes: $n_5$ and $n_6$. Hence, in principle, $x_1$ could be \emph{matched} to either of these nodes. Furthermore, the filter included in this clause specifies that the location of $x_1$ must be Barcelona, a condition satisfied by both nodes.  

The second clause includes an \emph{edge pattern} specifying that we are now searching for edges $y$ of type $byTrain$ from $x_1$ to some node $x_2$ of type $Airport$. In this case, the only possible edge that matches $y$ is $e_1$, with $x_1$ now corresponding to $n_6$ (so $n_5$ would be discarded as a match for $x_1$) and $x_2$ corresponding to $n_5$. Additionally, the filter specified in this clause states that the location of $x_2$ must be Barcelona, which $n_5$ satisfies. 

Finally, the third clause includes a \emph{path pattern} specifying that we are searching for paths $p$ consisting of a (non-empty) sequence of edges of type $Flight$, whose source is $x_2$ and whose target is some node $x_3$, both of type $Airport$. In path patterns, the form of a path $p$ is specified by writing $p : \alpha$, where $\alpha$ is a regular expression over the set of edge types.  This means that a path $(n,e_1 \dots e_k,n')$ can match $p$ only if $(et_1 \dots et_k)$ belongs to the language defined by $\alpha$, where $et_1, \dots, et_k$ are the types of $e_1, \dots, e_k$, respectively. 

We may notice that, in general, there is an infinite number of paths that could  match this pattern, and, consequently, the number of answers to a query may also be infinite. This issue is addressed in standard approaches by allowing only paths that do not include repeated nodes (\emph{simple or acyclic paths}), or repeated edges (\emph{trails}), or by restricting the results to the shortest paths. In our case,  we may adopt the same semantics, but, in addition, our path properties allow us to ensure that the number of answers to a query is finite, for example, by limiting the length of the paths.  More precisely, in the third clause we state that $x_3$ must be an airport of Los Angeles, that the overall cost of the flight must be less than 1000 euros, and that the length of the path should not exceed 3. So, in this case, the number of paths satisfying these constraints would be finite. These two path properties $cost$ and $length$, are defined in Fig. \ref{fig:defprop} using path unfolding, as described in the Introduction.

In particular, in our example in Fig. \ref{fig:defprop}, we define three path properties, $cost$ and $length$ used in our query example, and $start$, which specifies the start time of a travel connection. As can be seen,  on the left side of the figure we have the definitions of these properties when the given path $p$ denotes a single flight, and on the right side we have the definitions for the case where $p$ denotes a path including more than one flight. 
Moreover, in the latter case, we may also notice that, in addition to the equations defining the properties, we have other constraints. For example,  $p'\!.start > y.arr + 90$ that  filters paths that include two consecutive connections if the difference between the departure time of one connection and the arrival time of the previous is not greater than 90 minutes. The other constraints $p'\!.length \geq 0$ and $p'\!.cost \geq 0$ are used to discard ("wrong") answers as soon as possible, as we will see in Example \ref{ex:example}.

\begin{figure}[t]
    \centering
  \begin{center} 
  $
\begin{array}{ll}
\mathbf{Property \ definitions} \\ 
\mathbf{Properties: }  \ \  length, cost, start   \\ 
\mathbf{case} \ p \ \ \mathbf{of} \\
\ \  \ \ \  (x \stackrel{y}{\edge} x'):   \ \  \ \ \   \ \  \ \ \  \ \  \ \ \      \ \ \ \ \  \ \ \  \ \  \ \ \  \ \  \ \  \ \ \  (x \stackrel{y}{\edge} x'' \stackrel{p'}{\pathto} x'): \\
\ \  \ \ \ \ \  \ \ \    p.length == 1  \   \ \ \ \  \ \ \ \ \  \ \  \   \ \ \ \ \  \ \  \ \ \ \ \  \ \  \ \ \   p.length == 1+p'.length  \\
\ \ \ \  \ \ \   \ \ \  p.cost == y.price \ \  \ \ \   \ \ \ \ \  \ \  \ \ \ \ \  \ \ \ \ \ \     \ \   p.cost == y.price+p'.cost \\
\ \ \ \  \ \ \   \ \ \   p.start == y.dep \ \ \  \ \  \ \ \   \ \ \ \ \  \ \  \ \ \   \ \ \ \ \ \ \  \   p.start == y.dep  \\
\ \  \ \ \   \ \ \ \ \  \ \  \ \ \   \ \ \ \ \  \ \  \ \ \   \ \ \ \ \  \ \ \ \ \ \ \  \ \ \ \ \ \ \ \ \ \  \ \ \    \ \ \ \ \ \ \   \  \ \ \   \  \ \   p'.length > 0  \\
\ \  \ \ \   \ \ \ \ \  \ \  \ \ \   \ \ \ \ \  \ \  \ \ \   \ \ \ \ \ \ \ \ \ \  \ \ \ \  \ \ \  \ \ \ \ \ \ \  \ \ \   \ \ \ \ \ \ \    \ \ \    \ \    p'.cost > 0   \\
\ \  \ \ \   \ \ \ \ \  \ \  \ \ \  \ \  \ \ \   \ \ \ \ \  \ \  \ \ \   \ \ \ \ \ \ \ \ \ \  \ \ \ \  \ \ \  \ \ \ \ \  \ \ \ \ \ \ \    \ \ \    \ \   p'.start > y.arr + 90 \\
 \end{array}
$
\end{center}
    \caption{Definitions of path properties}
\label{fig:defprop}\label{fig:deflength} \label{fig:propconn}
\end{figure}

To conclude, we consider that a solution or answer to a query like this is a \emph{match} (i.e., a function mapping each node, edge or path variable to a node, edge or path in the GDB, respectively) that satisfies the associated filters. For example, an answer to the query above would be the match $\{x_1\mapsto n_6, x_2 \mapsto n_5, x_3 \mapsto n_1, y \mapsto e_1, p \mapsto e_6 e_7 \}$. However, a match mapping $p$ to $e_6 e_8$ would not be valid, since its cost would be 1050€, violating the constraint $cost < 1000$. Similarly, a match mapping $p$ to $e_5 e_3 e_2$ would also be invalid, as it would violate the constraint $p'.start > y.arr + 90$, given that the arrival time of $e_5$ is 11:30 and the departure of $e_3$ is 12:00, less than the required 90-minute interval.

\section{Preliminaries}
\label{sect:prelim}
In this section, we formally introduce the basic elements described in our running example. In particular, at the model level, we define property graphs and paths; then, at the syntax level, we define patterns, clauses, and queries; finally, we define matchings of patterns onto the database, establishing a first connection between both levels. 

\subsection{The data model}
\label{sec:datamodel} 

To define property graphs, we assume a given family of data domains $\D = \{\D_{vt}\}_{vt \in VT}$, indexed by a set $VT$ of value types (\emph{sorts}), including all data values that can be stored in the database as properties of nodes or edges. For example, $VT$ typically includes the types $int, real$, and $string$, and may also include others such as $time$. Thus, for instance, $\D_{int}$ denotes the set of integers. We also assume that there is a family of operations over these data types, such as $+$ or $*$, and predicates, such as $\le$, defined over integers or reals.  Additionally, we assume that an equality predicate $==$ is defined for all domains.
Finally, we assume that we have four pairwise disjoint countable sets: $\N$ representing nodes; $\E$ representing edges; $\K$ representing property identifiers; and $\Tt$ representing (type) labels. Then:

\begin{defi}[Property graph]
Given  a family of data domains $\D = \{\D_{vt}\}_{vt \in VT}$, a \emph{property graph} $G$ over $\mathcal{D}$ is a tuple 
$G = (N, E, P, src, tgt, \rho, \tau)$,  
where:
\begin{itemize}
\item $N \subset \N$ is a finite set of nodes;
\item $E\subset \E$ is a finite set of  edges;
\item $P  \subset \K$ is a finite set of property identifiers;
\item $src, tgt: E \to N$ are functions that map each edge to its source and target nodes, respectively.  
\item $\rho: (N \cup E)\times P \to  \mathcal{D}$ is a partial function that associates a value from $\mathcal{D}$ to each pair consisting of a node or edge and a property. 
\item $\tau: N \cup E \to 2^\Tt$ is a function that maps  each node or edge to a (possibly empty) set of labels.

\end{itemize}
\end{defi}

Notice that, according to this definition, each node or edge is associated with a set of labels. This means that they may have multiple types or, alternatively, no type at all if the associated set is empty. 

As usual, a path in a graph is a sequence of edges such that, for any two consecutive edges $e_1, e_2$  we have $src(e_2) = tgt(e_1)$. However, for our convenience, in general, we will consider that the origin and destination nodes of a path are also included as part of the path.

\begin{defi}[Paths]
A \emph{path} $\pi$ in a graph $G = (N, E, P, src, tgt, \rho, \tau)$ is a tuple, $\pi = (n, e_1\dots e_k, n')$, with $1\le k$, such that $n,n' \in N$ and $e_1, \dots, e_k \in E$, satisfying $src(e_1) = n$, $tgt(e_k) = n'$, and for each $j: 1<j\le k$, $src(e_j) = tgt(e_{j-1})$.
 
 For simplicity, a path may also be denoted just by its sequence of edges, $e_1\dots e_k$, when the source and target nodes are clear from the context.
 \end{defi}
 
 Notice that according to this definition we only consider nonempty paths. 
 In what follows,  $Paths(G)$ will denote the set of all (nonempty) paths in $G$. Moreover, for any path $\pi = (n, e_1\dots e_k, n')$, we write $src(\pi)$ to refer to $n$ and $tgt(\pi)$ to refer to $n'$.


\subsection{The query language}
\label{sec:datamodel}\label{sec:conds}
We first introduce the notion of \emph{filters}. To this end, we assume the existence of four pairwise disjoint, countable sets of variables: $XV$ for \emph{value variables}, typically denoted $z, z', z_1, \dots$; $XN$ for \emph{node variables}, typically denoted $x, x', x_1, \dots$; $XE$ for \emph{edge variables}, typically denoted $y, y', y_1, \dots$; and $XP$ for \emph{path variables}, typically denoted $p, p', p_1, \dots$. Furthermore, $\V$ will denote the set of all variables, that is, $\V = XV \cup XN \cup XE \cup XP$

\begin{defi}[Terms, atoms, filters, and variable substitutions]
 \emph{Terms} are inductively defined as follows:
 \begin{itemize}
    \item Values in $\D$, value variables and property expressions of the form $x.pr$ are terms, where $x$ is a node, edge, or path variable, and $pr$ is a property identifier. 
    \item If $t_1,\dots, t_k$ are terms and $op$ is an operation symbol, then $op(t_1, \dots, t_k)$ is also a term for $k \ge 0$.
     \end{itemize}
\emph{Atoms}   are expressions of the form $p(t_1, \dots, t_k)$,  for $k \ge 0$, where $p$ is a predicate symbol, and $t_1,\dots, t_k$ are terms. A special kind of atoms is an expression of the form $t_1 = = t_2$ (equivalently, $==(t_1,t_2))$, where the symbol $==$ denotes the \emph{equality predicate}. These atoms are often referred to as \emph{equations.}

\emph{Filters} are defined inductively as follows: 

 \begin{itemize}
    \item  Atoms are  filters. 
      \item If $\Phi$ and $\Phi'$ are filters, then both $\Phi \wedge \Phi'$ and $\neg \Phi$ are filters.
\item Finally, given sets of variables $X_1,X_2 \subseteq \V$, and a filter (or in general any first-order formula) $\Phi$, a \emph{variable substitution} is a mapping $m: X_1 \to X_2$. In this context, $m(\Phi)$ denotes the formula obtained by replacing in $\Phi$ any occurrence of a variable $x \in X_1$ with the variable $m(x)$.
       \end{itemize}
   
\end{defi}

Sometimes we identify conjunctions of filters with sets of filters, that is, we consider that $\Phi_1 \wedge \dots, \wedge \Phi_k$ as equivalent to the set $\{\Phi_1, \dots,  \Phi_k\}$.


As explained in Section \ref{running-example}, a clause consists of a pattern and a filter, and queries are sequences of clauses, where patterns specify (fragments of) graphs. In particular, a pattern describes the nodes, edges, and paths that should be included in the fragment. In our query language, we only use three kinds of patterns; however, other query languages, such as GQL \cite{GQL23, LibkinGPC}, may incorporate more complex pattern forms. However, in this paper, we also use patterns of the form $(x_1 \stackrel{y}{\edge} x_2 \stackrel{p:\alpha}{\pathto} x_3)$ in the context of path unfolding, as well as patterns in which some node or edge variables have been replaced by actual nodes and edges.

\begin{defi}[Patterns, clauses and queries]\label{def:queries}
A \emph{clause}  $\ell$ is a pair,   $\ell= \aP\Box\Phi$, where $\Phi$ is a  filter and $\aP$ is a \emph{pattern} of one of the following forms:
\begin{itemize}
    \item \emph{node pattern} $\aP = (x\!:\!nt)$, or
    \item \emph{edge pattern} $\aP = (x_1\!:\!nt_1 \stackrel{y:et}{\edge} x_2\!:\!nt_2)$, or
    \item \emph{path pattern} $\aP = (x_1\!:\!nt_1 \stackrel{p:\alpha}{\pathto} x_2\!:\!nt_2)$    
     \end{itemize}
where $x, x_1, x_2$ are node variables, $nt, nt_1, nt_2$ are their corresponding types; $y$ is an edge variable and $et$ is its type; finally, $p$ is a path variable and $\alpha$ is a regular expression over the alphabet of edge types, e.g., $Flight^+$.

Then, a \emph{query} $q$ is a sequence of clauses $\ell_1, \dots, \ell_k$, for $k>0$.

\end{defi}
  
As we can see, we assume that the variables in patterns are typed. In the case of nodes and edges, their type is a (possibly empty) set of type labels. In the case of path variables, we consider that their type is a regular expression over edge type labels. 
     For simplicity and readability, we write variables whose type is the empty set as if they were untyped. For example,  $(x_1\stackrel{y:et}{\edge} x_2)$ instead of $ (x_1\!:\emptyset \stackrel{y:et}{\edge} x_2\!:\emptyset)$.  Similarly, in the case of singletons, we omit set brackets. For example,  
 $ (x\!:\!Airport)$ instead of $(x\!:\!\{Airport\})$.

\subsection{Matchings}
\label{ss:matchings}

As usual, matchings are functions that map the variables of a pattern to the elements of the given graph (nodes, edges, or paths). However, we distinguish between two kinds of matchings, depending on whether path variables are matched or not. 

Let $Var(\aP)$ denote the set of node and edge variables occurring in a pattern $\aP$, and let $GVar(\aP)$ denote the set of node, edge, and path variables occurring in $\aP$. Furthermore, if $\ell = \aP\Box\Phi$, we define $Var(\ell)= Var(\aP)\cup Var(\Phi)$ and $GVar(\ell)= GVar(\aP)\cup GVar(\Phi)$, where $Var(\Phi)$ and $GVar(\Phi)$ denote, respectively, the sets of node and edge variables, and the set of node, edge, and path variables that occur in $\Phi$. Finally, if $q= \ell_1, \dots, \ell_k$, then $Var(q)$ and $GVar(q)$ will denote, respectively, the sets $Var( \ell_1)\cup\dots Var( \ell_k)$ and $GVar( \ell_1)\cup\dots GVar( \ell_k)$.

\begin{defi}[Match, general match]
A \emph{general match} $m$ of a pattern $\aP$ to a graph $G = (N, E, P, src, tgt, \rho, \tau)$ is a function $m: GVar(\aP) \to (N \cup E \cup Paths(G))$ (from now on, for simplicity, $m: GVar(\aP) \to G$) satisfying:
\begin{enumerate}
    \item[(1)] If $\aP$ is a node pattern $(x:nt)$ then $nt \subseteq \tau(m(x))$.
    \item[(2)] If $\aP = (x_1\!:\!nt_1 \stackrel{y:et}{\edge} x_2\!:\!nt_2)$,  then $nt_1 \subseteq \tau(m(x_1))$, $nt_2 \subseteq \tau(m(x_2))$, $et \subseteq \tau(m(y))$, $src(m(y)) = m(x_1)$, and $tgt(m(y)) = m(x_2)$.
    \item[(3)] If $\aP = (x_1\!:\!nt_1 \stackrel{p:\alpha}{\pathto} x_2\!:\!nt_2)$ and $m(p) = (n_1,e_1 \dots e_k,n_2)$, then $nt_1 \subseteq \tau(m(x_1))$, $nt_2 \subseteq \tau(m(x_2))$, $n_1 = m(x_1)$, $n_2 = m(x_2)$, and the sequence $\tau(m(e_1)) \dots \tau(m(e_k))$ is in the language defined by $\alpha$.

     \end{enumerate}
        A \emph{match} of $\aP$ to $G$ is a function $m: Var(\aP) \to G$  satisfying conditions (1), (2) and 
        \begin{enumerate}
    \item[(3')] If $\aP = (x_1\!:\!nt_1 \stackrel{p:\alpha}{\pathto} x_2\!:\!nt_2)$, then $nt_1 \subseteq \tau(m(x_1))$ and $nt_2 \subseteq \tau(m(x_2))$.

     \end{enumerate}

    \end{defi}
    
    Although it is an abuse of notation, given a set of variables $X$ such that $GVar(\aP) \subseteq X$, we will also call a function $m:  X \to G$ a general match of $\aP$ to $G$, if $m|_{GVar(\aP)}$ is a general match of $\aP$ to $G$, where $m|_{GVar(\aP)}$ denotes the function $m$ restricted to variables in $GVar(\aP)$. Similarly, we will also say that $m$ is a match  of $\aP$ to $G$, if $m|_{Var(\aP)}$ is a match of $\aP$ to $G$.

\section{Path properties}
\label{def:pproperties}

In this section, we describe how properties for paths are defined and how they are computed when evaluating a query. First, in Section \ref{ss: regular}, we introduce the notion of path unfolding. Then, in Section \ref{ss: DBSchemas}, we formally define how path properties are specified. Section \ref{ss: compunfold} is dedicated to presenting how we can associate a set of constraints with each unfolding of a given path pattern to compute its properties. Finally, in Section \ref{ss: compdef}, we show how to associate a set of constraints with a path representing the values of its properties, so that their values can be determined for a given answer.

 \subsection{Disjunctive decomposition of regular expressions and path unfolding}
\label{ss: regular} \label{sec:unfolding}

As mentioned above, unfolding a path pattern $\aP = (x\!:\!nt \stackrel{p:\alpha}{\pathto} x'\!:\!nt')$ means decomposing it into several path patterns that are, in some sense, equivalent to $\aP$. In particular, some patterns would have the form $\aP =(x\!:\!nt\stackrel{y:et}{\edge} x'\!:\!nt)$, and the remaining patterns in the decomposition would have the form $(x\!:\!nt \stackrel{y:et}{\edge} x'' \stackrel{p':\alpha'}{\pathto} x'\!:\!nt')$, where $y, x'', p'$ are fresh new variables, i.e., variables not used in the current context, and where $et$ and $\alpha'$ are derived from $\alpha$.

For example, if $\alpha = Flight^+$, then its unfolding would consist of $(x\!:\!nt\stackrel{y:Flight}{\edge} x'\!:\!nt)$ and $(x\!:\!nt \stackrel{y:Flight}{\edge} x'' \stackrel{p':Flight^+}{\pathto} x'\!:\!nt')$. However, if $\alpha = Flight^++byTrain^+$, then the unfolding of $(x\!:\!nt \stackrel{p:Flight^++byTrain^+}{\pathto} x'\!:\!nt')$ would consist of the following four patterns: 
\begin{itemize}
    \item $(x\!:\!nt\stackrel{y:Flight}{\edge} x'\!:\!nt)$,
    \item $(x\!:\!nt\stackrel{y:byTrain}{\edge} x'\!:\!nt)$,
    \item $(x\!:\!nt \stackrel{y:Flight}{\edge} x'' \stackrel{p':Flight^+}{\pathto} x'\!:\!nt')$, and
    \item $(x\!:\!nt \stackrel{y':byTrain}{\edge} x''' \stackrel{p'':byTrain^+}{\pathto} x'\!:\!nt')$.
\end{itemize}

In what follows, we will describe how the unfolding of a path pattern can be defined systematically based on the decomposition of its associated regular expression.

We assume the reader is familiar with regular expressions over an alphabet $\A$. In particular,  we consider regular expressions $\alpha$ of the form:
$$ \alpha =  a | \alpha ^* | \alpha ^+ | \alpha_1 + \alpha_2| \alpha_1 \alpha_2$$
where $a\in \A$. 

Note that we intentionally exclude $\varepsilon$, the empty sequence, from this definition. This is because, as stated earlier, we assume all paths in our framework are non-empty.

Let $\R(\A)$ denote the set of all regular expressions over the alphabet $\A$. For any $\alpha \in \R(\A)$, we write $\Le(\alpha)$ to denote the language defined by $\alpha$, that is, the set of non-empty sequences $a_1 \cdots a_n \in \A^+$ that match $\alpha$. Then, given the regular expressions $\alpha$ and $\alpha'$, we say that $\alpha'$ is \emph{more general} than $\alpha$, written $\alpha \le \alpha'$, if $\Le(\alpha) \subseteq \Le(\alpha')$ and we say that are equivalent, written $\alpha \equiv \alpha'$, if $\alpha \le \alpha'$ and $\alpha' \le \alpha$. In particular, $\A^+$ is the most general regular expression, as it represents the set of all non-empty sequences over $\A$.

  \emph{Disjunctive decomposition} is a property of regular expressions that allows us to decompose an expression $\alpha$ into an \emph{equivalent}  set of  expressions. This decomposition plays a key role in defining the \emph{unfolding} of path patterns. The property says that for every $\alpha \in \R(\A)$,  there there exist two sets $S_0,S_1\!\subseteq\!\mathcal{A}$ and a \emph{remainder} function
$rem: S_1 \to \R(\A)$, such that 
$\Le(\alpha) = \bigcup_{a\!\in\!S_0} \{a\} \cup \bigcup_{b\!\in\!S_1} b\cdot\Le(\!rem(b))$; 
where $b \cdot \Le(rem(b))$ denotes the set of all sequences obtained by concatenating $b$ with any sequence from $\Le(rem(b))$. 

For example, given $\alpha=a+ab^++ac^++c$, we obtain its disjunctive decomposition with $S_0 = \{a,c\}$, $S_1 = \{a\}$, and $rem(a) = b^++c^+$. Hence:   
$$\Le(a+ab^++ac^++c) = \{a\} \cup  \{c\} \cup a\!\cdot\!\Le(b^++c^+) $$
Thus, by iterating this decomposition, we can obtain all sequences in $\Le(\alpha)$ of length 1, 2, and so forth.

Disjunctive decompositions are finite, provided that $\mathcal{A}$ is finite and they are unique up to equivalence of the remainders.  Specifically, if $(S_0, S_1, rem)$ and $(S'_0,S'_1, rem')$ are two decompositions of $\alpha$,  then $S_0 = S'_0$, $S_1 = S'_1$ and for each $a \in S_1: rem(a) \equiv rem'(a)$, that is, $\Le(rem(a)) = \Le(rem'(a))$.

To unfold a path pattern $\aP = (x\!:\!nt \stackrel{p:\alpha}{\pathto} x'\!:\!nt')$, we generally apply the disjunctive decomposition of $\alpha$ as defined above. However, there is a special case where a simplified, yet equivalent, decomposition is preferred for reasons of conciseness and readability. In particular, given the pattern $(x\!:\!nt \stackrel{p:\A^+}{\pathto} x'\!:\!nt')$, where $\A = \{et_1, \dots, et_k\}$, following the general case (see Def. below), its general unfolding would consist of the following set of patterns  $$(x\!:\!nt\stackrel{y:et_1}{\edge} x'\!:\!nt), \dots, (x\!:\!nt\stackrel{y:et_k}{\edge} x'\!:\!nt), (x\!:\!nt \stackrel{y:et_1}{\edge} x'' \stackrel{p':\A^+}{\pathto} x'\!:\!nt'), \dots, (x\!:\!nt \stackrel{y:et_k}{\edge} x'' \stackrel{p':\A^+}{\pathto} x'\!:\!nt')$$ 
Instead, we adopt the simplified unfolding for $(x\!:\!nt \stackrel{p:\A^+}{\pathto} x'\!:\!nt')$, now denoted by $(x\!:\!nt \stackrel{p}{\pathto} x'\!:\!nt')$, consisting of just two patterns:  $(x\!:\!nt\stackrel{y}{\edge} x'\!:\!nt')$ and 
$(x\!:\!nt \stackrel{y}{\edge} x'' \stackrel{p'}{\pathto} x'\!:\!nt')$.


\begin{defi}[Unfolding of path patterns]
Given a pattern $\aP = (x\!:\!nt \stackrel{p:\alpha}{\pathto} x'\!:\!nt')$, its \emph{unfolding set} $\U(\aP)$ is:

\begin{enumerate}
\item If $\alpha  \equiv \A^+$, then $\U(\aP) = \{(x\!:\!nt\stackrel{y}{\edge} x'\!:\!nt), 
(x\!:\!nt \stackrel{y}{\edge} x'' \stackrel{p'}{\pathto} x'\!:\!nt')\}$.
\item Otherwise, 
if $ (S_0,S_1,rem)$ is the disjunctive decomposition of $\alpha$, $\U(\aP)$ is the set consisting of the following patterns:
\begin{itemize}
\item  $(x\!:\!nt \stackrel{y:et}{\edge}  x'\!:\!nt')$,  for each $et\in S_0$

 \item  $(x\!:\!nt \stackrel{y:et}{\edge} x'' \stackrel{p':rem(et)}{\pathto} x'\!:\!nt')$,  for each $et\in S_1$.
\end{itemize}
\end{enumerate}
where $y, x'',p'$ are fresh new variables.

\noindent Sometimes we will refer to $p$ as the \emph{unfolded variable}.

\end{defi}

\subsection{Defining path properties}
\label{ss: DBSchemas}
Assigning properties to paths involves specifying how these properties are to be computed. 
As illustrated in the example in Fig. \ref{fig:defprop}, this is achieved by associating with each pattern in the unfolding of the given path pattern  a set of constraints. These constraints define both the computation of the path’s properties and, where appropriate, additional conditions that the paths must satisfy.

In principle, a constraint may be any kind of first-order formula, provided that an efficient constraint solver is available. However, we believe that in most cases, it is enough to use (conjunctions of) linear equalities and inequalities, as illustrated in the example presented in this paper. We denote the set of all possible constraints by $Constr$.

\begin{defi}[Path property definitions]\label{def:pprop}
A \emph{path property definition}, $\dpp$ is a 4-tuple 
$\dpp = (PP, p, \U, \Delta)$, where:
\begin{itemize}
    \item $PP  \subseteq \K$ is a set of path property identifiers.
    \item The unfolding set $\U= \{(x\stackrel{y}{\edge} x'),
(x\stackrel{y}{\edge} x'' \stackrel{p'}{\pathto} x')\}$ of the pattern $(x\stackrel{p}{\pathto} x')$.

    \item $p$ is the path variable unfolded by $\U$.

    \item A function $\Delta: \U  \to 2^{Constr}$ mapping each pattern $\aP\in \U$ to a set of constraints (or, equivalently, a conjunction of constraints). 
 \end{itemize}
 \end{defi}

For example, $\Delta(x \stackrel{y}{\edge} x'' \stackrel{p'}{\pathto} x') $ in Fig. \ref{fig:defprop}  includes equations defining the given path properties, such as $p.length == 1 + p'.length$, but also additional constraints like $p'.start > y.arr + 90$, which filters out some unwanted answers, or $p'.length > 0$ and $p'.cost > 0$ that allow us to prune the search space, as illustrated in Example\ref{ex:example}.

\subsection{Constraints associated with a path pattern unfolding}
\label{ss: compunfold}
In this section, given a path property definition $\dpp = (PP, p, \U, \Delta)$, and a path pattern $\aP = (x_1\!:\!nt_1 \stackrel{p_1:\alpha_1}{\pathto} x_2\!:\!nt_2)$, we define the set of constraints associated with each unfolding of $\aP$, which will be used to compute the properties of $p$.

The idea is quite simple. For each unfolding of $\aP$ of the form $(x_1\!:\!nt_1 \stackrel{y_1:e_1}{\edge} x_2\!:\!nt_2)$, we associate the set $m(\Delta(x\stackrel{y}{\edge} x'))$, where $m$ is the variable substitution   that replaces the variables $x,x',y$, and $p$, with $x_1,x_2,y_1$, and $p_1$, respectively. Similarly, for each unfolding of $\aP$ of the form $(x_1\!:\!nt_1 \stackrel{y_1:e_1}{\edge} x_3 \stackrel{p_2:\alpha_2}{\pathto} x_2\!:\!nt_2)$, we associate the set $m(\Delta(x\stackrel{y}{\edge} x''\stackrel{p'}{\pathto} x' ))$, where $m$ is the variable substitution that replaces the variables $x,x',x'',y,p'$, and $p$, with $x_1,x_2,x_3,y_1,p_2$, and $p_1$, respectively.

\begin{defi}[Constraints for path patterns]\label{def:unfoldconstr}
If $\dpp = (PP, p, \U, \Delta)$ is a path property definition and $\aP = (x_1\!:\!nt_1 \stackrel{p_1:\alpha_1}{\pathto} x_2\!:\!nt_2)$ is a path pattern, we define the set of \emph{constraints associated with each unfolding} $u$ of $\aP$, 
$ConstrPU(\aP,u,Def)$ as follows: 
\begin{itemize}
\item If $u$ is of the form $(x_1\!:\!nt_1 \stackrel{y_1:e_1}{\edge} x_2\!:\!nt_2)$, then $ConstrPU(\aP,u,Def)= m(\Delta(x\stackrel{y}{\edge} x'))$, where $m$ is the variable substitution $m = \{m(x) \mapsto x_1, m(x') \mapsto x_2, m(y) \mapsto y_1, m(p) \mapsto p_1\}$.

\item If $u$ is of the form $(x_1\!:\!nt_1 \stackrel{y_1:e_1}{\edge} x_3 \stackrel{p_2:\alpha_2}{\pathto} x_2\!:\!nt_2)$, then $ConstrPU(\aP,u,Def)= m(\Delta(x\stackrel{y}{\edge} x'')\stackrel{p'}{\pathto} x'))$, where $m$ is the variable substitution $m = \{m(x) \mapsto x_1, m(x'') \mapsto x_3, m(x') \mapsto x_2. m(y) \mapsto y_1, m(p) \mapsto p_1, m(p') \mapsto p_2\}$.

\end{itemize}
\end{defi}

For example, in the case of the property definitions of our running example and the path pattern $\aP = (x_1\!:\!Airport \stackrel{p_1:Flight^+}{\pathto} x_2\!:Airport)$, the set of constraints associated with the unfolding $u_1= (x_1\!:Airport \stackrel{y_1:Flight}{\edge} x_2\!:\!Airport)$ is $$ConstrPU(\aP,u_1,Def) = \{p_1.length == 1, p_1.cost == y_1.price,  p_1.start == y_1.dep\}.$$

Similarly, the set of constraints associated with the unfolding $u_2 = (x_1\!:\!nt_1 \stackrel{y_1:e_1}{\edge} x_3 \stackrel{p_2:\alpha_2}{\pathto} x_2\!:\!nt_2)$ is:
\begin{align}
ConstrPU(\aP,u_2,Def) = &\{p_1.length == p_2.length+1, p_1.cost == p_2.cost+y_1.price,  \nonumber \\ &p_1.start == y_1.dep, p_2.length > 0, p_2.cost > 0, p_2.start > y_.arr + 90 \}\nonumber.
\end{align}

\subsection{Constraints associated with a path}
\label{ss: compdef}

As we have seen, in our example, we define path properties using equations that look like recursive definitions in some programming languages. However, their semantics is different. The reason is shown in the following example. Suppose that $p$ is a path variable that denotes a sequence of flights from Barcelona to Los Angeles that must satisfy the filter $p.length  \le 2$. Since we know that this sequence cannot consist of a single flight, because there are no direct flights from Barcelona to Los Angeles, we decide that it must consist of a flight from Barcelona to some airport $n$, followed by a sequence of flights $p'$ from $n$ to Los Angeles. However, if we want to check if $p.length  \le 2$, we would need to compute $p.length == 1 + p'.length$, but this is not possible because at this point we have no idea of what the path denoted by $p'$ is or its length. 

In our framework, the equations defining path properties are understood purely as declarative constraints rather than procedural computations. For instance, in the example where a path $p$ consists of an edge $e$ followed by a subpath $p'$, the constraint $p.\text{length} = 1 + p'.\text{length}$ 
simply states that this equation must hold, independently of which path is $p'$. 
The key issue is that when we eventually know which path $\pi$ is denoted by $p$, we expect that the constraints associated with $\pi$ will tell us the values of its path properties. In particular, we consider that each path has an associated set of constraints that can be computed by \emph{unfolding the path}, and by associating to each unfolding step the corresponding set of constraints.


\begin{defi}[Constraints for a path]\label{def:pathconstr}
Given a path $\pi =(n_1,e_1 \dots e_k,n_2)$ and a path property definition $Def = (PP, p, \U=\{u_1,u_2\}, \Delta)$, where  $u_1 = (x\stackrel{y}{\edge} x')$ and  $u_2=(x \stackrel{y}{\edge} x'' \stackrel{p'}{\pathto} x')$, we
define  inductively the \emph{set of constraints associated with} $\pi$, $Constr(Def,\pi)$,  as  follows:
\begin{itemize}
\item if $k = 1$: 
$$Constr(Def,\pi) = m(\Delta(u_1))$$
where $m$ is a general match of $GVar(u_1)\cup\{p\}$, defined $\{m(x) \mapsto n_1, m(x') \mapsto n_2, m(y) \mapsto  e_1, m(p) \mapsto \pi \}$, and $m(\Delta(u_1))$ denotes the set of constraints in $\Delta(u_1)$, where each occurrence of $x, x'$, and $y$   have been replaced with $n_1, n_2$, and $e_1$, respectively.

\item If $k >1$:
$$Constr(Def,\pi) = m(\Delta(u_2))\cup Constr(Def,\pi')$$
 where  $\pi' =(src(e_2),e_2 \dots e_k,n_2)$ and $m$ is a general match of $GVar(u_2)\cup\{p\}$, defined $\{m(x) \mapsto n_1, m(x') \mapsto n_2, m(x'') \mapsto src(e_2), m(y) \mapsto e_1, m(p') \mapsto  \pi', m(p) \mapsto  \pi$\}.   
 
\end{itemize}
 \end{defi}

\begin{exa}
Let $\pi=(n_5, e_6 e_7, n_1)$  be a path in Fig.~\ref{fig:running-example},  $Constr(Def,\pi)$ can be computed as follows.   Since  the length of $\pi$ is not $1$ ($\pi$ consists of $e_6$ followed by the path
$\pi'=(n_4, e_7, n_1)$), the set of constraints associated to  $\pi$ would be $Constr(Def,\pi) = m(\Delta(u_2))\cup Constr(Def,\pi')$, i.e., 

\begin{align}
Constr(Def,\pi)  = &m\big(\{p.length == 1+p'.length, p.cost == y.price+p'.cost,  p.start == y.dep,  \nonumber \\  &p'.length > 0, p'.cost > 0, p'.start > y.arr + 90\}  \cup Constr(Def,\pi')\big) \nonumber
\end{align}
where $m$ is the general match defined $\{m(x) \mapsto n_5,  m(x') \mapsto n_1, m(x'') \mapsto n_4,
m(y) \mapsto e_6, m(p) \mapsto \pi, m(p') \mapsto \pi'\}$. Hence,

\begin{align}
Constr(Def,\pi)  = &\{\pi.length == 1+\pi'.length, \pi.cost == e_6.price+\pi'.cost,  \pi.start == e_6.dep, \nonumber \\  &\pi'.length > 0, \pi'.cost > 0, \pi'.start > e_6.arr + 90\}  \cup Constr(Def,\pi')\big)\nonumber
\end{align}
which is equivalent to 
\begin{align}
\{&\pi.length == 1+\pi'.length, \pi.cost == 650+\pi'.cost,  \pi.start == \text{9:00}, \pi'.length > 0, \nonumber \\  &\pi'.cost > 0, \pi'.start > \text{15:00} + 90\}\cup Constr(Def,\pi')\nonumber
\end{align}

Now, since  the length of $\pi'$ is $1$, the set of constraints associated to  $\pi'$ is 
$$m(\Delta(u_1)) = m\big(\{p.length == 1, p.cost == y.price,  p.start == y.dep\} \big)$$
where $m$ is the general match defined $\{m(x) \mapsto  n_4,m(x') \mapsto  n_1, m(y) \mapsto  e_7, m(p) \mapsto  \pi'\}$. Therefore: 
$$Constr(Def,\pi') =\{\pi'.length == 1, \pi'.cost == e_7.price,  \pi'.start == e_7.dep\}  $$ 
which is equivalent to 
$$Constr(Def,\pi') = \{\pi'.length == 1, \pi'.cost == 300,  \pi'.start == \text{17:00}\} $$

As a consequence, $Constr(Def,\pi)$  is the set 
 \begin{align}
\{&\pi.length == 1+\pi'.length, \pi.cost == 650+\pi'.cost,  \pi.start == \text{9:00}, \pi'.length > 0,  \nonumber \\ &\pi'.cost > 0, \pi'.start > \text{15:00} + 90\}\cup \{\pi'.length == 1,  \pi'.cost == 300,  \pi'.start == \text{17:00}\} \nonumber
\end{align}
or equivalently:
 \begin{align}
\{&\pi.length == 2, \pi.cost == 950,  \pi.start == \text{9:00}, \pi'.length == 1, \pi'.cost == 300,  \nonumber \\ &\pi'.start == \text{17:00},\text{17:00}> \text{15:00}+ 90\}\nonumber
\end{align}

 Thus, from $Constr(Def,\pi)$ we can infer that $\pi.lengh==2$, $\pi.cost==950$, and $\pi.start==\text{17:00}$, i.e., the values of the properties of $\pi$.  Moreover, we may notice that this set of constraints is satisfiable. In particular, this implies that there is a gap of more than 90 minutes between the departure time of the flight represented by the edge $e_7$ and the arrival time of the flight represented by $e_6$.

 \end{exa}

However, in general, given a path $\pi$, $Constr(Def,\pi)$  may not necessarily yield definite values for its properties. This is because a set of constraints might not have a unique solution: $Constr(Def, \pi)$ could either have no solutions or admit several solutions. In our framework, as we may see in our operational semantics, as soon as we detect that a set of constraints is unsatisfiable, that possible solution is immediately discarded.  
But the latter case may occur. For instance, in the previous example , if $e_7$ did not have the property $price$, the set of constraints associated with $\pi$ would include the equations $\{\pi.cost ==   650 +\pi'.cost,  ~ \pi'.cost ==   e_7.price\}$. After simplification, we would obtain
$\{ \pi.cost ==   650 + e_7.price\}$, and we would not be able to assign a concrete value to the cost of that path. 

This is not necessarily a problem. Path properties are primarily used to filter out solutions that we consider inadequate for specific reasons. Thus, in this case, we would not discard the path as a possible solution, because it is not proven to be inadequate. If, in addition, we would like to return to the user the value of all path properties, including $cost$, we would consider that the answer should be $\pi.cost == 650 + e_7.price$,
 as done in Constraint Logic Programming (CLP) (see, e.g. \cite{JaffarM94,JaffarMMS98}), where we would consider that the value of a path property is not necessarily a value, but a set of simplified constraints.

In any case, we will write $Constr(Def,\pi) \vdash \pi.pr == v$, if the constraints that define the property $pr$ of a path $\pi$ have a unique solution, that is, if $v$ is the only value that a property $pr$ can have for a path $\pi$ so that $Constr(Def,\pi)$ is satisfiable.

The following results will be used later.

\begin{lem}\label{lem:constr1}
Given a path $\pi =(n_1,e_1 \dots e_k,n_2)$ and a path property definition $Def = (PP, p, \U=\{u_1,u_2\}, \Delta)$, where  $u_1 = (x\stackrel{y}{\edge} x')$ and  $u_2=(x \stackrel{y}{\edge} x'' \stackrel{p'}{\pathto} x')$,
then for any $i:1\le i < k$,  we have:
$$Constr(Def,\pi) = \bigcup_{1\le j \le i} m_j(\Delta(u_2)) \cup Constr(Def,\pi_{i+1})$$
where $m_j$  is the general match $m_j: GVar(u_2)\cup\{p\} \to G$  defined  
$\{m_j(x)\mapsto src(e_j), m_j(x') \mapsto n_2, m_j(x'') \mapsto tgt(e_j), m_j(y) \mapsto e_j,  m_j(p)\mapsto \pi_j, m_j(p')\mapsto \pi_{j+1}\}$ and 
$\pi_j =(src(e_{j}),e_{j} \dots e_k,n_2)$, for any $j$.
\end{lem}

\begin{proof}
By induction on $i:1\le i < k$:
\begin{itemize}
\item The case $i=1$ is trivial, since the resulting equality is just the definition of $Constr$.

\item If $i>1$, by the induction hypothesis we have:
$$Constr(Def,\pi) = \bigcup_{1\le j < i} m_j(\Delta(u_2)) \cup Constr(Def,\pi_{i})$$
Now, according to the definition of Constr, we have:
 $$Constr(Def,\pi_{i}) = m_i(\Delta(u_2))\cup Constr(Def,\pi_{i+1})$$
 where  $m_i: GVar(u_2)\cup\{p\} \to G$  is defined 
 $\{m_i(x)\mapsto src(e_i), m_i(x') \mapsto n_2, m_i(x'') \mapsto tgt(e_i), m_i(y) \mapsto e_i,  m_i(p)\mapsto \pi_i, m_i(p')\mapsto \pi_{i+1}\}$ and $\pi_{i+1}=(src(e_{i+1}),e_{i+1} \dots e_k,n_2)$. Therefore:
 $$Constr(Def,\pi) = \bigcup_{1\le j \le i} m_j(\Delta(u_2)) \cup Constr(Def,\pi_i)$$ for each i, $i:1\le i < k$.
\end{itemize}
\end{proof}

\begin{lem}\label{lem:constr2}
Given a path $\pi =(n_1,e_1 \dots e_k,n_2)$ and a path property definition $Def = (PP, p, \U=\{u_1,u_2\}, \Delta)$, where  $u_1 = (x\stackrel{y}{\edge} x')$ and  $u_2=(x \stackrel{y}{\edge} x'' \stackrel{p'}{\pathto} x')$:
\begin{enumerate}
\item For  $m_k: GVar(u_1)\cup\{p\} \to G$,  defined  
$\{m_k(x)\mapsto src(e_k), m_k(x')\mapsto n_2, m_k(y)\mapsto e_k, m_k(p) \mapsto (src(e_k), e_k, n_2\}$, we have
$m_k(\Delta(u_1)) \subseteq Constr(Def,\pi)$.

\item For any $i:1\le i <k$, if $\pi_{i} =(src(e_{i}),e_{i} \dots e_k,n_2)$, then for $m_i: GVar(u_2)\cup\{p\} \to G$,  defined  
$\{m_i(x)\mapsto src(e_i), m_i(x') \mapsto n_2, m_i(x'') \mapsto tgt(e_i), m_i(y) \mapsto e_i,  m_i(p)\mapsto \pi_i,  m_i(p')\mapsto \pi_{i+1}\}$, we have
$m_i(\Delta(u_2)))\subseteq Constr(Def,\pi)$.   
 
\end{enumerate}
\end{lem}

\begin{proof}
\begin{enumerate}
\item If $\pi_k = (src(e_{k}),e_k,n_2)$, according to the definition of $Constr$,
$Constr(Def,\pi_k) = m_k(\Delta(u_1))$, where $m_k: GVar(u_2)\cup\{p\}  \to G$, is defined  
$\{m_k(x)\mapsto src(e_k), m_k(x')\mapsto n_2, m_k(y)\mapsto e_k\}, m_k(p) \mapsto (\pi_k\}$. Therefore,  by Lemma \ref{lem:constr1},
$m_k(\Delta(u_1)) \subseteq Constr(Def,\pi)$.

\item Direct consequence of Lemma \ref{lem:constr1}.   
 
\end{enumerate}

\end{proof}

 \section{Logical semantics of queries}
\label{sec:queries}
\label{def:matching}\label{def:sat-queries}\label{ss: sat}

In this section, we assume that clauses and queries are formulas of a certain logic whose models are graphs.  This will allow us both to define a simple, formal semantics for our query language, referred to as the \emph{logical semantics}, and to verify the adequacy of our operational semantics defined in the following section.  Specifically, adequacy is established by proving two properties: soundness and completeness. Soundness ensures that every answer produced by the operational semantics is valid according to the logical semantics, while completeness ensures that all logically valid answers can be derived through the operational semantics. In this way, the logical semantics serves as a specification for the query language, and the operational semantics acts as an abstract implementation. Thus checking soundness and completeness in this context is like verifying the implementation of the language. 
   
    Let $m$ be any kind of match, or more generally, any function $m: V\to G$, where $V$ is a set of variables. Although this may be considered an abuse of notation,  given a term $t$,  a  filter $\Phi$,   a pattern $\aP$, a clause $\ell$, or a  query $q$,     we will denote by $m(t)$, $m(\Phi)$, $m(\aP)$,  $m(\ell)$, or $m(q)$, the result of replacing in $t$, $\Phi$, $\aP$, $\ell$, or $q$, respectively, every occurrence of a variable $x$ in $V$ by $m(x)$.

Roughly speaking, a clause is satisfied under a given (general) match if the conditions associated with that clause are consistent. However, these conditions depend on the kind of clause.

\begin{defi}[Conditions of a clause]
Given a graph $G$, a set of path property definitions, $Def$, a clause $\ell =  \aP\Box \Phi$, and a general match $m: GVar(\ell) \to G$, the set of \emph{conditions  associated with} $\ell$ and $m$, $Cond(\ell, m, Def)$, is defined: 
 \begin{itemize}
\item   If $\ell=\aP\Box \Phi$, where  $\aP = (x\!:\!nt)$ or $\aP =(x_1\!:\!nt_1 \stackrel{y:et}{\edge} x_2\!:\!nt_2)$, we have  $Cond(\ell, m, Def)= m(\Phi)$.


\item   If $\ell= \aP\Box \Phi$, with $\aP = (x_1\!:\!nt_1 \stackrel{p:\alpha}{\pathto} x_2\!:\!nt_2)$, we have
$Cond(\ell, m, Def)= m(\Phi)\cup m(Constr(Def, m(p)))$.
\end{itemize}
\end{defi}

Now we can define our notion of satisfaction.

\begin{defi}[Clause and query satisfaction]\label{def:sat}
Given a graph $G$,  a set of path property definitions $Def$, and a general match $m: GVar(\aP) \to G$,   we say that $G$\emph{ satisfies a clause}  $\ell=\aP\Box \Phi$ in the context of $Def$, with respect to $m$, denoted 
 $(G,Def, m)  \models  \ell$ if $Cond(\ell,m,Def)$ is satisfiable.
 
And we say that $G$ \emph{satisfies  a query} $q= \ell_1, \dots, \ell_k$, in the context of $Def$, with respect to the general match $m: GVar(q) \to G$, denoted 
$(G,Def,m)  \models  q$, if  $Cond(\ell_1,m,Def) \cup\dots \cup Cond(\ell_k,m,Def)$ is satisfiable.

Finally,  we say that $m$ is a \emph{correct answer} for $q$, if $(G,Def, m)  \models  q$. 

\end{defi}

\section{Operational semantics}
\label{subsec:OPSem} \label{ssec:comp-ans}

Our operational semantics is defined by three derivation rules that describe how query answers are computed. Roughly speaking, to obtain these answers, starting from an \emph{initial state}, we successively apply these rules, transforming the state step by step, until we reach a \emph{final state} that, in a well-defined sense, includes the answer to the query.  

In this section,  we first define our notion of a state, and then describe in detail our semantics.


\subsection{States}\label{def:states}
States include two kinds of information: what remains to be done, i.e., which clauses have yet to be solved to compute the answer to a given query and what information we have gathered in the process of reaching that state, including part of the matchings that will provide the answer to the query. 

\begin{defi}[States]
Given a graph database $\GDB$, a \emph{state} $\sigma$ on $G$ is a 4-tuple 
$\sigma= (q,\Psi,\M,\Pi)$,
where: 
\begin{enumerate}
\item[(a)] $q$ is the \emph{current query} to be solved, that is, the
 sequence of clauses not yet solved.
\item[(b)] $\Psi$ is the set of constraints and conditions from the already solved clauses, including the constraints that define the path properties.
\item[(c)] $\M$ is the match that we have already computed for the node and edge variables. More precisely, $\M$ is a set of pairs $x \mapsto a$, where $x$ is a node or edge variable and $a$ is a node or edge in $\GDB$, respectively.
\item[(d)] $\Pi$ is the match that we have partially computed for the path variables. More precisely, $\Pi$ is a set of triples $p \mapsto (s, p')$, where $p$ is the name of the path variable, $s$ is a sequence of edges in $\GDB$, and $p'$ is a path variable or the symbol $\lambda$. In particular,  the path eventually assigned to $p$ will be the concatenation of $s$ and the path eventually assigned to $p'$, if  $p' \neq \lambda$; otherwise it will be $s$.
\end{enumerate}
The \emph{initial state} is $\sigma_{0} = 
(q_0, \emptyset, \emptyset, \emptyset)$, where $q_0$ is the given query, and a \emph{final state}  is a state $\sigma= (\emptyset,\Psi,\M,\Pi)$.
\end{defi}

\subsection{Derivation rules}\label{def:deriv-rules}

When we are in a state $\sigma= (q,\Psi,\M,\Pi)$, the computation of the next state consists roughly of two steps:  a)  selecting a clause from $q$ and b) finding a match for the clause and updating the components of the state accordingly. This means that we will not try to solve the clauses following their order in the query, that is, we consider that $q$ is a set of clauses, rather than a sequence.

In standard query languages, such as Cypher \cite{CYPHER2018} or GQL \cite{GQL23, LibkinGPC}, to avoid non-termination, paths are asked to be simple paths (no repeated nodes), trails (no repeated edges), or shortest paths. In our framework, we do not impose these restrictions, as non-termination can be avoided by using some path properties, like length. Nevertheless, there would be no problem asking the paths to be simple or trails since our results would also hold.

After applying a rule, we must be sure that the new state is consistent: 

\begin{defi}[State consistency]
 A state $\sigma= (q,\Psi,\M,\Pi)$ is \emph{consistent}  if  
\begin{enumerate}
\item $\Psi$ is satisfiable;
\item if  $\{x \mapsto n,  x \mapsto n'\} \subseteq \M$, then $n=n'$.
\item if $\{p\mapsto (s,p'), p\mapsto (s',\lambda)\} \subseteq \Pi$, then $s$ is a prefix of $s'$. Otherwise, if  $\{p\mapsto (s,p'), p\mapsto (s',p'')\}\subseteq \Pi$ and $p',p''$ are different from $\lambda$, then either $s$ is a prefix of $s'$ or $s'$ is a prefix of $s$.
\item if $\{p\mapsto (s,\lambda), p\mapsto (s',\lambda)\} \subseteq \Pi$, then $s = s'$. 
\end{enumerate}
\end{defi}

We can now define our derivation rules:

\begin{defi}[Derivation rules]
\label{def:rules}
Given a database $G$, a path property definition $\dpp = (PP, p, \U, \Delta)$, where the variables occurring in $\U$, i.e., $x,x',x'', y, p, p'$ have been renamed apart from the rest of the variables of the query, and a state  $\sigma= (q,\Psi,\M,\Pi)$, where  $q = \{\ell\} \cup  q'$ and  $\ell$ is a \emph{selected clause}, we  can derive a new state $\sigma'$, written  $ \sigma \deriv_{G,Def} \sigma'$, if we can apply one of the following rules:
\begin{enumerate}
\item[\tt(R)]\label{R1} If $\ell = \aP\Box\Phi$, with $\aP = (x_1\!:\!nt_1)$ or $\aP =(x_1\!:\!nt_1 \stackrel{y_1:et_1}{\to} x_2\!:\!nt_2)$,  and there is
a match $m\!:Var(\aP)\!\to\!\GDB$, then: 
$$((\{\ell\} \cup q'),\Psi, \M, \Pi)~~ \deriv_{\GDB,Def} ~~ \sigma' = (q',\Psi',\M', \Pi)$$
if $\sigma'$ is consistent, where $\Psi' = m(\Psi \cup \Phi)$ and $\M' = \M \cup \{x\mapsto m(x)\mid x\in  Var(\aP)\}$
\item[\tt(U1)]\label{U1} 
If $\ell = (x_1\!:\!nt_1\!\stackrel{p_1:\alpha}{\pathto}\!x_2\!:\!nt_2)\Box\Phi$, 
$u=(x_1\!:\!nt_1\!\stackrel{y_1:et_1}{\edge}\!x_2\!:\!nt_2)$  
is the chosen unfolding, and there
 is a match $m: Var(u)\to G$,  then
$$((\{\ell\}\cup q'),\Psi, \M, \Pi) ~~ \deriv_{\GDB,Def} ~~ \sigma'=(q',\Psi' , \M', \Pi')$$
if $\sigma'$ is consistent, and where:
\begin{itemize}
\item $\M' = \M \cup \{x_1 \mapsto m(x_1), x_2 \mapsto m(x_2), y_1 \mapsto m(y_1)\}$; 
\item $\Psi' = m(\Psi \cup \Phi) \cup m(ConstrPU(\aP,u,Def))$; 
\item If there is a triple 
$(p_0 \mapsto(s,p_1)) \in \Pi$,    then
$\Pi' = \{(p \mapsto (s,p'))\in \Pi \mid p \ne p_0\} \cup \{p_0 \mapsto (s \cdot m(y_1),\lambda)\}$. Otherwise, 
$\Pi' =  \Pi \cup \{p_1 \mapsto (m(y_1),\lambda)\}$.
  \end{itemize}
\item[\tt(U2)]\label{U2} 
If $\ell = (x_1\!:\!nt_1\!\stackrel{p_1:\alpha_1}{\pathto}\!x_2\!:\!nt_2)\Box\Phi$,   the chosen unfolding is
$u= (x_1\!:\!nt_1\!\stackrel{y_1:et_1}{\edge}\!x_3 \stackrel{p_2:\alpha_2}{\pathto}\!x_2\!:\!nt_2)$, and
there is a match $m\!:\!Var(\aP_0)\!\to\!\GDB$, where $\aP_0 = (x_1\!:\!nt_1\!\stackrel{y:et}{\edge}\!x_3)$, then
$$((\{\ell\}\cup q'),\Psi, \M, \Pi) ~~ \deriv_{\GDB,Def} ~~\sigma'=  (\{\ell'\} \cup q',\Psi' , \M', \Pi')$$
if $\sigma'$ is consistent, and where:
\begin{itemize}
\item  $\ell' = (x_3 \stackrel{p_2:\alpha_2}{\pathto}\!x_2\!:\!nt_2)\Box\Phi$; 
\item $\M' = \M \cup \{x \mapsto m(x)\mid x\in  Var(\aP_0)\}$;
\item $\Psi' =  m(\Psi \cup \Phi) \cup m(ConstrPU(\aP,u,Def))$.
 
\item  If there is a triple $(p_0 \mapsto(s,p_1)) \in \Pi$, then
$\Pi' = \{p \mapsto (s,p')\in \Pi \mid p \ne p_0\} \cup \{p_0 \mapsto (s \cdot m(y_1),p_2)\}$. Otherwise, 
$\Pi' =  \Pi \cup \{p_1 \mapsto (m(y_1),p_2)\}$. 

\end{itemize}

\end{enumerate} 
\end{defi}

Therefore, if $\ell$ includes a node or edge pattern, rule {\tt (R)} tries to match the variables in $\aP$ to nodes or edges in the database,   so that the existing constraints are satisfiable, and if it succeeds, it adds the new match to $\M$,  the new constraints to $\Psi$, and removes the clause from $q$.  

If $\ell=(x_1\!:\!nt_1\!\stackrel{p_1:\alpha_1}{\pathto}\!x_2\!:\!nt_2)\Box\Phi$ and the unfolding chosen is $u= (x_1\!:\!nt_1\!\stackrel{y_1:et_1}{\edge}\!x_3 \stackrel{p_2:\alpha_2}{\pathto}\!x_2\!:\!nt_2)$, rule {\tt (U2)} starts proceeding as in rule {\tt (R)} with clause $(x_1\!:\!nt_1\!\stackrel{y_1:et_1}{\edge}\!x_3) \Box\Phi$ and match $m$.  In addition, it replaces $\ell$ in $q$ by $(x_3 \stackrel{p_2:\alpha_2}{\pathto}\!x_2\!:\!nt_2)\Box\Phi$, it adds to $\Psi$ all the constraints in $\Phi$, plus the constraints associated with the unfolding $u$ according to $Def$.  Finally, if there is a $p_0 \mapsto(s,p_1)$ in $\Pi$, then $p_1$ was not a path variable in $q$, but was introduced by some unfolding of $p_0$. So, it replaces the triple in $\Pi$  by $p_0 \mapsto (s\cdot m(y_1), p_2)$. But
if there is no $p_0 \mapsto(s,p_1)$  in $\Pi$, then $p_1$ is a path variable in $q$ which we have not tried to solve yet, so it adds $p_1 \mapsto (m(y_1),p_2)$ to $\Pi$. 

The case where the chosen unfolding is $u=(x_1\!:\!nt_1\!\stackrel{y_1:et_1}{\edge}\!x_1\!:\!nt_1)$ is similar (using rule {\tt (U1)}), but simpler.

\begin{exa}
\label{ex:example}
Suppose that we want to get answers from the database depicted in Fig.\ref{fig:PPG-Instance}, for the query $(x_1\!:\!Airport\!\stackrel{p_1:Flight^+}{\pathto}\!x_2\!:\!Airport)\Box \Phi$, where $\Phi =\{x_1.loc == \hbox{``Barcelona''}, x_2.loc == \hbox{``Los Angeles''}, p_1.length\le 2\}$. Since the clause is a path pattern, we will have to apply  rule {\tt(U1)} or {\tt(U2)}. We cannot choose {\tt(U1)}, since  there are no direct flights from Barcelona to Los Angeles.
So, the only unfolding to choose is $(x_1:Airport\!\stackrel{y_1:Flight}{\edge}x_3 \stackrel{p_2:Flight^+}{\pathto}\!x_2\!:\!Airport)$. So we have to apply rule {\tt(U2)}. Hence, we have to find a match from $(x_1:Airport\!\stackrel{y_1:Flight}{\edge}x_3)$ to G and replace the original clause with $\ell_1 = (x_3 \stackrel{p_2:Flight^+}{\pathto}\!x_2\!:\!Airport)\Box \Phi$. There are two possible matches: $m_1 = (m_1(x_1) \mapsto n_5,  m_1(x_3) \mapsto n_3, m_1(y_1) \mapsto e_5)$ and $m_2 = (m_2(x_1) \mapsto n_5,  m_2(x_3) \mapsto n_4, m_2(y_1) \mapsto e_6)$. Suppose that we choose $m_1$  the resulting state would be $\sigma_1= (\ell_1, \Psi_1, \M_1, \Pi_1)$, where $\Psi_1= m_1(\Phi) \cup m'(\Delta(u_2))$,  and $m'$ is defined as in rule {\tt (U2)}.
So, $\Psi_1=\{n_5.loc == \hbox{``Barcelona''}, x_2.loc == \hbox{``Los Angeles''}, p_1.length\le 2 \}\cup \{p_1.length == 1+p_2.length, p_2.length > 0, p_1.cost == 150+p_2.cost, p_1.start ==\hbox{10:00},p_2.start > \text{11:30}+ 90, p_2.cost \ge 0\}$, $\M_1=\{ x_1 \mapsto n_5,  x_3 \mapsto n_3, y_1 \mapsto e_5\}$, and $\Pi_1=\{ p_1 \mapsto (e_5,p_2)\}$.

Now, to solve $\ell_1$, we can only choose the unfolding: $(x_3 \stackrel{y_2:Flight}{\edge}\!x_4 \stackrel{p_3:Flight^+}{\pathto}\!x_2\!:\!Airport)$, since again there are no direct flights from Paris to Los Angeles. So we apply again rule {\tt(U2)}, i.e., we have to  find a match from $(x_3 \stackrel{y_2:Flight}{\edge}\!x_4)\Box \Phi$ and replace clause $\ell_1$ by $\ell_2= (x_4 \stackrel{p_3:Flight^+}{\pathto}\!x_2\!:\!Airport)\Box \Phi$ in the new state.   Here, the only possible match would be $m = (m(x_3) \mapsto n_3,  m(x_4) \mapsto n_2, m(y_2) \mapsto e_3)$. However,  the new state would be inconsistent. The reason is that the set of constraints $\Psi_2$ of the new state would include the constraints $p_2.start > \text{11:30}+ 90$ and $p_2.start ==\text{12:00}$, which together are unsatisfiable.

However, suppose that the departure time of $e_3$ is 14:30, then, even in that case, the new state will be inconsistent. The reason is that $\Psi_2$ would include the constraints $\{p_1.length  == 1 + p_2, p_2.length = 1+ p_3.length, p_1.length \le 2, p_3.length > 0 \}$.
Notice that if the constraint $p'.length > 0$ had not been included in $\Delta(x\stackrel{y}{\edge}x'' \stackrel{p'}{\pathto}x')$, $p_3.length > 0$ would not have been included in $\Psi_3$ and this inconsistency would not have been detected. As a consequence, instead of detecting the failure, we would have continued to search if some extension of the path $e_5 e_3 $ could be a correct answer.

So, if we were implementing the operational semantics in terms of a search procedure, we would need to backtrack and choose the match $m_2$, in the first step. Then, after two more steps, we would have found that $e_6 e_7$ is a correct match for $p$ that solves the query. 

\end{exa}

Let us now see what answer is computed by a \emph{successful} derivation.

\begin{defi}[Derivations and computed answers]\label{def:computedansw}
A \emph{derivation} $d$ in \GDB ~for a query $q$ is a sequence of states  obtained by the  application of the above rules, starting from the initial state:
$$d = \sigma_0 ~ \deriv_{\GDB,Def} ~ \dots ~ \deriv_{\GDB,Def} ~ \sigma_i $$
A derivation is \emph{successful} if  it ends  at a final state $\sigma$ of the form $(\emptyset,\Psi, \M, \Pi)$. In this case, the \emph{computed answer} obtained from that derivation is the general match $m_d$, defined for each $x \in Var(q)$:
\begin{enumerate}

\item If $(x \mapsto a)\in \M$, ~~$m_d(x) = a$.
\item Similarly, if $p \mapsto (e_1\dots e_k,\lambda)\in \Pi$,  ~~$m_d(p) = (src(e_1),e_1\dots e_k, tgt(e_k))$.

\end{enumerate}
\end{defi}

This definition is correct, as Proposition \ref{prop:def} shows, but before we prove a simple lemma:

\begin{lem}\label{lemma:Psi}
For any query $q_0$ and any derivation $\sigma_0 = (q_0,\emptyset,\emptyset,\emptyset) \!\stackrel{*}{\deriv}_{\GDB} \sigma_k =\!(q_k,\Psi_k,\!\M_k,\!\Pi_k)$, and any  $i: 0\le i \le k$ we have:

\begin{enumerate}
 \item [(1)]  Every node or edge variable $x \in Var(q_i)$ is either in $Var(q_k)$ or there is a pair $x\mapsto v \in \M_k$.
 \item [(2)]  For every path variable $p \in GVar(q_i)\setminus GVar(q_0)$, $p\in GVar(q_k)$ iff there is a triple $p'\mapsto (s,p) \in \Pi_k$, for some $p'$.
 
  \item [(3)]\label{lem:caso3}  For every path variable $p \in GVar(q_0) $, $p$ is either in $GVar(q_k)$ or there is a triple $p\mapsto (s,p') \in \Pi_k$, for some $p'$.
 
\item[(4)]    $\sigma_k$ is consistent.

\end{enumerate}
\end{lem}

\begin{proof}
We proceed by induction on the length of the derivation. If $\sigma_0 \!\stackrel{*}{\deriv}_{\GDB}\!\sigma_k$ has $0$ steps, the four statements trivially hold. 

Let us suppose that  $\sigma_0 \stackrel{*}{\deriv}_{\GDB}  \sigma_k$ has $n+1$ steps, i.e.   
$$\sigma_0 \stackrel{n}{\deriv}_{\GDB} 
\sigma_{n} = (q_{n}, \Psi_{n}, \M_{n}, \Pi_{n})  \deriv_{\GDB}\sigma_{k}$$ 

\begin{enumerate}
\item[(1)]   Let $x\in Var(q_i)$ be a node or edge variable. If $i=k$, the property trivially holds. If $i<k$, by the inductive hypothesis, $x\in Var(q_n)$ or there is a pair $x\mapsto v \in \M_n$. In the latter case $x\mapsto v \in \M_k$, since by definition, $\M_j \subseteq \M_{j+1}$, for every $j$. Hence, we have to prove that, if $x$ is a node or edge variable in $q_n$ and there is no pair $x\mapsto v \in \M_k$, then $x$ is in $q_{k}$. We have the following cases, depending on the rule applied in the derivation step $\sigma_n\deriv_{\GDB}\sigma_{n+1}$:

\begin{itemize} 

\item If the rule applied in the derivation step
 $\sigma_n  \deriv_{\GDB}\sigma_{k}$  is {\tt(R)}  and  the chosen clause is $\ell=\aP\Box\Phi$, then $$\M_{k} = \M_n \cup \{y \mapsto g(y) \mid y \hbox{ is a node or edge variable in } \aP\}.$$ On the other hand, if $x$ is a node or edge variable in  $Var(q_{n}) \setminus Var(q_{k})$, this means that $x \in Var(\aP)$, according to the definition of the application of {\tt(R1)}. Therefore, $x\mapsto v \in \M_{n+1}$.

\item The cases where the rule applied  is {\tt(U1)} or {\tt(U2)} are similar to the case of {\tt(R1)}.

\end{itemize}

\item[(2)]  Let $p\notin q_0$ be a path variable. We have the following cases, depending on the rule applied in the derivation step $\sigma_n\deriv_{\GDB}\sigma_{n+1}$:
\begin{itemize}
    \item If the rule applied in the derivation step $\sigma_n  \deriv_{\GDB}\sigma_{n+1}$ is {\tt(R)}, then the property trivially holds, since this rule does not add or delete any path variable, nor it modifies $\Pi$. So, by the inductive hypothesis, $p\in GVar(q_n)$ iff there is a triple $p'\mapsto (s,p) \in \Pi_n$, which means that $p\in GVar(q_k)$ iff there is a triple $p'\mapsto (s,p) \in \Pi_k$.

 \item If the rule applied in the derivation step is {\tt(U1)} and the selected clause is $(x_1\!:\!nt_1\!\stackrel{p:\alpha}{\pathto}\!x_2\!:\!nt_2)\Box\Phi$, then $p_1 \notin GVar(q_k)$, because path variables that are not in $GVar(q_0)$ can only occur once in a derived query, but by definition, the triple $p'\mapsto (s,p_1)$ would be replaced by the triple $p'\mapsto (s,\lambda)$. So, there would not be any triple of the form $p'\mapsto (s,p_1)$ in $\Pi_k$. On the other hand, if $p \neq p_1$, we have $p \in GVar(q_k)$ iff $p \in GVar(q_n)$, because no path variable different from $p_1$ is deleted from $GVar(q_n)$. But, by the inductive hypothesis, $p \in GVar(q_n)$ iff there is a triple $p'\mapsto (s,p) \in \Pi_n$. Finally, there is a triple $p'\mapsto (s,p) \in \Pi_n$ iff there is a triple $p'\mapsto (s,p) \in \Pi_k$, since no triple $p'\mapsto (s,p_1)$ is removed from $\Pi_k$, if $p \neq p_1$.

\item If the rule applied in the derivation step is {\tt(U2)} and the selected clause is $\ell = (x_1\!:\!nt_1\!\stackrel{p_1:\alpha_1}{\pathto}\!x_2\!:\!nt_2)\Box\Phi$, that is, $q_n = \ell \cup q'_n$, and $q_k = \ell' \cup q'_n$, where $\ell' = (x_3 \stackrel{p_2:\alpha_2}{\pathto}\!x_2\!:\!nt_2)\Box\Phi$ then, by definition, $p_1 \notin GVar(q_k)$  and, as in the previous case, there would not be any triple of the form $p'\mapsto (s,p_1)$ in $\Pi_k$. On the other hand, $p_2 \in GVar(q_k)$ and, by definition, $\Pi_k = \{p_0 \mapsto (s,p)\in \Pi \mid p \ne p_1\} \cup \{p' \mapsto (s \cdot m(y_1),p_2)\}$. As a consequence, like in the previous case, by the inductive hypothesis, the property holds.

\end{itemize}

\item[(3)] If $p \in GVar(q_0)$, by the inductive hypothesis, we may assume that $p \in GVar(q_{n})$ or there is a triple $p\mapsto (s,p') \in \Pi_{n}$, for some $p'$. We have the following cases, depending on the rule applied in the derivation step $\sigma_n\deriv_{\GDB}\sigma_{n+1}$:
\begin{itemize}
    \item If the rule applied in the derivation step $\sigma_n  \deriv_{\GDB}\sigma_{n+1}$ is {\tt(R)}, then if the property holds for  $\sigma_n$ it will also hold for $\sigma_k$, since this rule does not add or delete any path variable, nor it modifies $\Pi$. 
\item If the rule applied in the derivation step $\sigma_{k-1} \deriv_{\GDB} \sigma_k$ is {\tt(U1)} and the selected clause is $\ell = (x_1\!:\!nt_1\!\stackrel{p_1:\alpha_1}{\pathto}\!x_2\!:\!nt_2)\Box\Phi$. Then, if there is a triple $p\mapsto (s,p') \in \Pi_{n}$, then there would be a triple $p\mapsto (s,p'') \in \Pi_{k}$, because a triple $p\mapsto (s,p') \in \Pi_{n}$ can, at most be replaced by another triple $p\mapsto (s,p'') \in \Pi_{k}$. Finally, by definition, we know that $p_1$ is the only path variable that is in $GVar(q_n)$ but not in $GVar(q_k)$, therefore, if $p_1 \neq p$ then $p \in GVar(q_k)$. Conversely, if $p = p_1$, by definition, $\Pi_k$ would include a triple $p_1\mapsto (s,p_2) \in \Pi_{k}$.

\item If the rule applied is {\tt(U2)} and the selected clause is $\ell = (x_1\!:\!nt_1\!\stackrel{p_1:\alpha_1}{\pathto}\!x_2\!:\!nt_2)\Box\Phi$. Then, like in the previous case, if there is a triple $p\mapsto (s,p') \in \Pi_{n}$, then there would be a triple $p\mapsto (s,p'') \in \Pi_{k}$.

Finally, similarly to the previous case, we know that $p_1$ is the only path variable that is in $GVar(q_0)$ and in $GVar(q_n)$ but not in $GVar(q_k)$, therefore, if $p_1 \neq p$ then $p \in GVar(q_k)$. Conversely, if $p = p_1$, by definition, $\Pi_k$ would include a triple $p_1\mapsto (s,p_2) \in \Pi_{k}$.

\end{itemize}

\item[(4)] According to Def. \ref{def:rules}, the resulting state after the application of a rule must always be consistent.
 
\end{enumerate}
\end{proof}

\begin{prop}\label{prop:def}
 Definition \ref{def:computedansw} is correct, in the sense that, if ~$\sigma_0\stackrel{*}{\deriv}_{\GDB}   (\emptyset,\Psi,\M, \Pi)$ is a successful derivation, its computed answer $m$ is well defined on every node, edge, or path variable $x$ in $q_0$. 
 \end{prop}
 
\begin{proof}
We have to prove the following.
\begin{enumerate}
\item[a)] If $x$ is a node or edge variable in $q$, there is a pair $(x\mapsto v) \in \M$.
\item[b)]  If $(x\mapsto v),(x\mapsto v')  \in \M$, then $v= v'$.
\item[c)]  If $p$ is a path variable in $q$, there is a triple $p \mapsto (s,\lambda) \in \Pi$.
\item[d)]  If $p \mapsto (s,\lambda) , p \mapsto (s',\lambda)  \in \Pi$, then $s=s'$.
\end{enumerate}

Item a) is a direct consequence of Lemma \ref{lemma:Psi}, since if $x$ is a node or edge variable in $q$ and the sequence of clauses in a final state is empty, there must be a pair $(x\mapsto v) \in \M$.

Item b) is a consequence of the consistency of all states in a derivation.

For item c), by property (3) of Lemma \ref{lemma:Psi}, we know that there should be a triple $p \mapsto (s,p') \in \Pi$. On the other hand, by property (2) of the same lemma, $p'$ cannot be a path variable, since the sequence of clauses of a final state is empty. Hence, $p' = \lambda$.

Finally, item d), as item b), is a consequence of the consistency of all states in a derivation.
\end{proof}

\section{Correctness of the operational semantics}
\label{sec:sound-compl}
 
In this section, we show that our operational semantics is correct with respect to the logical semantics, in the sense that: a) it is \emph{sound}, i.e., every computed answer by a successful derivation $d = (q,\emptyset,\emptyset,\emptyset) \!\stackrel{*}{\deriv}_{\GDB,Def}\!(\emptyset,\Psi,\!\M,\!\Pi)$ is correct; b) it is complete, that is, every logical answer may be obtained by the operational semantics; and c)  if according to the logical semantics, a property of a path has a value, then the operational semantics will compute that value.   But before proving these results, we will first prove two supporting lemmas:

\begin{lem}\label{lem:lem3} %

Given a query $q_0 = \ell_1, \dots, \ell_j$ and a successful derivation $\sigma_0 = (q_0, \emptyset, \emptyset, \emptyset)  \deriv_{\GDB,Def}  \sigma_1 = (q_1, \Psi_1, \M_1, \Pi_1) \deriv_{\GDB,Def} \dots \deriv_{\GDB,Def}\sigma_k= (\emptyset, \Psi_k, \M_k, \Pi_k)$, for every literal $\ell=\aP\Box\Phi$ in the query $q_i$ of any state $\sigma_i$, for $0\le i \le k$, there exists a literal $\ell' = \aP'\Box\Phi'$ in $q_0$, such that $\Phi = \Phi'$. %
\end{lem}

\begin{proof}
By induction on $i:0\le i \le k$:
\begin{itemize}
\item If $i = 0$, the case is trivial.
\item  If $i >0$, we may assume that for any literal $\ell=\aP\Box\Phi$ $q_{i-1}$, there exists a literal $\ell' = \aP'\Box\Phi'$ in $q_0$, such that $\Phi = \Phi'$. Suppose that we have obtained $\sigma_i$ from $\sigma_{i-1}$ by applying a derivation rule to the literal $\ell=\aP\Box\Phi$. Now, if the rule applied  is {\tt(R)} or {\tt(U1)}, then $q_i \subseteq q_{i-1}$, so the property trivially holds for any literal in $q_i$.

But if the applied rule is {\tt(U2)}, then $q_i = (q_{i-1} \cup \ell'')  \setminus \ell$, where $\ell''= \aP''\Box \Phi$, for some pattern $\aP''$.  Then, by the inductive hypothesis, we know that there is a literal $\ell' = \aP'\Box\Phi'$ in $q_0$, such that $\Phi = \Phi'$.

\end{itemize}
\end{proof}

\begin{lem}\label{lem:lem2}
Given a query $q_0 = \ell_1, \dots, \ell_j$, such that $m$ is its computed answer for the derivation $\sigma_0 = (q_0, \emptyset, \emptyset, \emptyset)  \deriv_{\GDB,Def}  \sigma_1 = (q_1, \Psi_1, \M_1, \Pi_1) \deriv_{\GDB,Def} \dots \deriv_{\GDB,Def}\sigma_k= (\emptyset, \Psi_k, \M_k, \Pi_k)$, and given $\Psi = Cond(\ell_1,m,Def) \cup\dots \cup Cond(\ell_j,m,Def)$, we have $\Psi = \Psi_k $, up to variable renaming. 
\end{lem}

\begin{proof}
Let us   define a mapping  $m'$   from $\bigcup_{0 \le i \le k} GVar(q_i)$ to $\GDB$:
$$m'(x) = \left\{\begin{array}{lll}
		v	& \hbox{ if $x$ is a node or edge variable and there is a pair } \\ &x\mapsto v \in \M_k, \\
		s	& \hbox{ if $x$ is a path variable and there is a triple } \\&x \mapsto (s,\lambda) \in \Pi_k\\
		s'' & \hbox{ if $x$ is a path variable, there is a triple } \\&p \mapsto (s,x) \in \Pi_i 
		\hbox{ for some $i$, and } m(p)= s' \\
\end{array}\right.$$
where $s''$ is the sequence satisfying that $s' = s \cdot s''$. 

As a direct consequence of Lemma \ref{lemma:Psi}, $m'$ is well-defined. Moreover, notice that
$m'|_{GVar(q_0)} = m$, and that,
according to the definition of the derivation rules, for every $i: 0\le  i < k$, $m'(\Psi_i) \subseteq m'(\Psi_{i+1})$, which means that $m'(\Psi_i )\subseteq m'(\Psi_{k}) = \Psi_k$. 

We will first prove that $\Psi  \subseteq \Psi_{k}$ (up to variable renaming). 
More precisely, we will show that for every clause $\ell=\aP\Box\Phi$ in any state $\sigma_i$,  $0 < i \le k$, we have $ Cond(\ell,m',Def)\subseteq \Psi_k$. Notice that this implies $\Psi  \subseteq \Psi_{k}$, since, by definition,  $\Psi = Cond(\ell_1,m,Def) \cup\dots \cup Cond(\ell_j,m,Def)$. 

\begin{enumerate}
\item Suppose that $\ell = \aP\Box\Phi$, where $\aP$ is a node or an edge pattern. This means that, in some state $\sigma_{i'}$, {\tt(R)} would have been applied to $\ell$ with matching $m'|_{Var(\aP)}$. Therefore, by definition of {\tt(R)}, $\Psi_{i'+1} =m'|_{Var(\aP)}(\Psi_{i'} \cup \Phi)$, so  $m'(\Psi_{i'+1}) =m'(m'|_{Var(\aP)}(\Psi_{i'} \cup \Phi)) \subseteq \Psi_k$. But $m'(m'|_{Var(\aP)}(\Psi_{i'} \cup \Phi)) = m'(\Psi_{i'} \cup \Phi)$ and $Cond(\ell,m'|_{\aP},Def) = m'(\Phi))$. Hence $ Cond(\ell,m',Def)\subseteq \Psi_k$.
\item If $\ell = (x_1\!:\!nt_1 \stackrel{p_1:\alpha_1}{\pathto} x_2\!:\!nt_2)\Box\Phi$, we will prove by induction on the length of the path $m'(p_1)$, that $ Cond(\ell,m',Def)\subseteq \Psi_k$. 
\begin{itemize}
\item If the length of $m'(p_1)$ is 1, that is, $m'(p_1)= (n_1,e_1, n_2)$, where $y$ is an edge in $\GDB$ whose source and target are $n_1, n_2$, respectively, then, in some state $\sigma_{i'}$, {\tt(U1)} has been applied to $\ell$, with unfolding $u= (x_1\!:\!nt_1\!\stackrel{y:et}{\edge}\!x_2\!:\!nt_2)$
and match $m'|_{Var(u)}$.
 
 Then, by definition of {\tt(U1)},
 $\Psi_{i'+1} =m'|_{Var(u )}(\Psi_{i'} \cup \Phi) \cup m'|_{Var(u )}(ConstrPU(\aP,u,Def))$. 
 
 Hence, 
  $m'(\Psi_{i'+1}) =m'(m'|_{Var(u)}(\Psi_{i'} \cup \Phi)) \cup m'(m'|_{Var(u )}(ConstrPU(\aP,u,Def))) = m'(\Psi_{i'} \cup \Phi)\cup m'(ConstrPU(\aP,u,Def)) \subseteq \Psi_k$. But according to Def.~\ref{def:unfoldconstr}, $ConstrPU(\aP,u,Def)= m''(\Delta(x\stackrel{y}{\edge} x'))$, where $m''$ is the variable substitution $m'' = \{m''(x) \mapsto x_1, m''(x') \mapsto x_2, m''(y) \mapsto y_1, m''(p) \mapsto p_1\}$, and by
  Def.~\ref{def:pathconstr}, $Constr(Def,m'(p_1)) = m'(m''(\Delta(u_1)))$. Therefore,  $ Cond(\ell,m',Def) = m'(\Psi_{i'} \cup \Phi)\cup m'(m''(\Delta(u_1))) \subseteq \Psi_k$.

\item  If the length of $m'(p_1)$ is $j+1$, that is, $m'(p_1) = (m'(x_1) , e_1 \dots e_j e_{j+1}, m'(x_2))$, 
this means that, in some state $\sigma_{i'}$, {\tt(U2)} has been applied to $\ell$, with unfolding $u= (x_1\!:\!nt_1\!\stackrel{y:et}{\edge}\!x_3\!\stackrel{p_2:\alpha_2}{\pathto}\!x_2\!:\!nt_2)$
and matching $m'|_{Var(\aP_0)}$, where $\aP_0 = (x_1\!:\!nt_1\!\stackrel{y:et}{\edge}\!x_3)$. 

We have to prove 
  $$ Cond(\ell,m',Def) =  m'(\Phi)\cup m'(Constr(Def, m'(p_1))) \subseteq \Psi_k$$ 
  with 
  $$Constr(Def,m'(p_1)) = m'(m''(\Delta(u_2)))\cup Constr(Def,m'(p_2))$$
  where $u_2 = (x \stackrel{y}{\edge}\!x'' \stackrel{p'}{\pathto}\!x')$ and 
 $m''$ is the general match of $u_2\cup \{p\}$, defined $\{m''(x) \mapsto x_1, m''(x'') \mapsto x_3, m''(y) \mapsto y_1, m''(x')\mapsto x_2, m''(p) \mapsto p_1,  m''(p')\mapsto p_2 \}$. Notice that this means that $m''(\Delta(u_2)) = ConstrPU(\aP,u,Def)$. Hence,
 $$m'(Constr(Def,m'(p))) = m'(m'(ConstrPU(\aP,u,Def)))\cup m'(Constr(Def,m'(p_2)))$$
 but $m'(m'(ConstrPU(\aP,u,Def))) = m'(ConstrPU(\aP,u,Def))$ and $ m'(Constr(Def, m'(p_2)))=Constr(Def,m'(p_2))$.
 Therefore, we have to prove: 
 $$m'(\Phi)\cup m'(ConstrPU(\aP,u,Def))\cup Constr(Def,m'(p_2)) \subseteq \Psi_k$$ 
 
Now,  
$$\Psi_{i'+1} = m'|_{Var(\aP_0)}(\Psi \cup \Phi) \cup m'|_{Var(\aP_0)}(ConstrPU(\aP,u,Def))$$ 
 hence, 
  $m'(\Psi_{i'+1}) =m'(m'_{Var(\aP_0)}(\Psi \cup \Phi)) \cup m'(m'_{Var(\aP_0)}(ConstrPU(\aP,u,Def))) =$ \\ $ = m'(\Psi_{i'} \cup \Phi)\cup m'(ConstrPU(\aP,u,Def)) \subseteq \Psi_k$.
   
  This means that $m'(\Phi) \subseteq \Psi_k$ and $m'(ConstrPU(\aP,u,Def)) \subseteq \Psi_k$. So, it only remains to prove $ Constr(Def,m'(p_2)) \subseteq \Psi_k$. However, according to the definition of {\tt(U2)}, $q_{i'+1}$ includes the clause $\ell' = (x_3\!\stackrel{p_2:\alpha'}{\pathto}\!x_2\!:\!nt_2)\Box\Phi$, where the length of $m'(p_2)$ is smaller than the length of $m'(p_1)$. So, by the induction hypothesis, $ Cond(\ell',m',Def)  \subseteq \Psi_k$. but $Cond(\ell',m',Def) =  m'(\Phi)\cup m'(Constr(Def, m'(p_2)))$, which means that $m'(Constr(Def, m'(p_2)))\subseteq \Psi_k$.

\end{itemize}

\end{enumerate}

Let us now show that $\Psi_k \subseteq \Psi$. In particular, we will prove by induction that for every $i: 0\le  i \le  k$, $m'(\Psi_i)\subseteq  \Psi$.
   
   \begin{itemize}
   \item If $i=0$, trivially, $\Psi_i = \emptyset \subseteq \Psi$.
   \item If $i>0$ and $\ell$ is the chosen literal for the transition $\sigma_{i-1} \deriv_{\GDB,Def} \sigma_i$, we have three cases:
   \begin{enumerate}
   \item  If $\ell = \aP\Box\Phi$, where $\aP$ is a node or an edge pattern, then, by definition, $\Psi_i = m'|_{Var(\aP)}(\Psi_{i-1} \cup \Phi)$. So,
   $m'(Psi_i)=m'(m'|_{Var(\aP)}(\Psi_{i-1} \cup \Phi))=m'(\Psi_{i-1} \cup \Phi)$. But, on the one hand, by the induction hypothesis, $m'(\Psi_{i-1})  \subseteq \Psi$ and, on the other hand, by Lemma \ref{lem:lem3}
   there must exist a literal $\ell_{i'} = \aP_{i'}\Box\Phi_{i'}$ in $q_0$, such that $\Phi = \Phi_{i'}$, which means that $m'(\Phi)  \subseteq \Psi$.
      
   \item If $\ell = (x_1\!:\!nt_1 \stackrel{p_1:\alpha_1}{\pathto} x_2\!:\!nt_2)\Box\Phi$ and the applied rule is {\tt(U1)}, with unfolding $u= (x_1\!:\!nt_1\!\stackrel{y:et}{\edge}\!x_2\!:\!nt_2)$
and matching $m'|_{Var(u)}$, 
then $\Psi_i = m'|_{Var(u)}(\Psi_{i-1} \cup \Phi) \cup m'|_{u}(ConstrPU(\aP,u,Def))$. Hence, $m'(\Psi_i )= m'(m'|_{Var(u)}(\Psi_{i-1} \cup \Phi)) \cup m'(m'|_{Var(u)}(ConstrPU(\aP,u,Def))$ $= m'(\Psi_{i-1} \cup \Phi) \cup m'(ConstrPU(\aP,u,Def))$, By the inductive hypothesis, we have $m'(\Psi_{i-1})  \subseteq \Psi$, and by Lemma \ref{lem:lem3}, as in the previous case, 
 $m'(\Phi)  \subseteq \Psi$. So we just have to prove that $m'(ConstrPU(\aP,u,Def)) \subseteq \Psi$.

 According to Def.~\ref{def:unfoldconstr}, $ConstrPU(\aP,u,Def)= m''\Delta((u_1))$, where $u_1 = (x\stackrel{y}{\edge} x')$ and $m''$ is the variable substitution $m'' = \{m''(x) \mapsto x_1, m''(x') \mapsto x_2, m''(y) \mapsto y_1, m''(p) \mapsto p_1\}$, and according to Lemma~\ref{lem:caso3}.3, since $p\in GVar(q_i)$, either $p\in GVar(q_i)$ or there is a triple $p' \mapsto(s,p)\in \Pi_i$. In the former case, according to Definition \ref{def:pathconstr}, $Constr(Def,m'(p))=m'(m''(\Delta(u_1)))$, which means that $m'(m''(\Delta(u_1))) \subseteq \Psi$. Otherwise, if there is a triple $p' \mapsto(s,p)\in \Pi_i$, then according to Lemma \ref{lem:constr2}, 
 $m'(m''(\Delta(u_1))) \subseteq Constr(Def,m'(p'))$. Hence $m'(m''(\Delta(u_1))) \subseteq  \Psi$, since $p'\in GVar(q_0)$.

\item If $\ell = (x_1\!:\!nt_1 \stackrel{p_1:\alpha_1}{\pathto} x_2\!:\!nt_2)\Box\Phi$ and the rule applied is {\tt(U2)}, the proof is similar. In particular, $\Psi_i = m'|_{Var(\aP)}(\Psi_{i-1} \cup \Phi) \cup m'|_{Var(\aP)}(ConstrPU(\aP,u,Def))$.

Hence, $m'(\Psi_i )= m'(m'|_{Var(\aP)}(\Psi_{i-1} \cup \Phi)) \cup m'(m'|_{Var(\aP)}(ConstrPU(\aP,u,Def))$ $= m'(\Psi_{i-1} \cup \Phi) \cup m'(ConstrPU(\aP,u,Def))$. 
By the inductive hypothesis, we have $m'(\Psi_{i-1})  \subseteq \Psi$, by Def.~\ref{def:unfoldconstr}, $ConstrPU(\aP,u,Def)= m''(\Delta(u_2))$, where $u_2 = (x \stackrel{y}{\edge}\!x'' \stackrel{p'}{\pathto}\!x')$  and $m''$ is the variable substitution defined $\{m''(x) \mapsto x_1, m''(x'') \mapsto x_3, m''(y) \mapsto y_1, m''(x')\mapsto x_2, m''(p) \mapsto p_1,  m''(p') \mapsto p_2 \}$. 
and by Lemma \ref{lem:lem3}, as in the previous case 
 $m'(\Phi)  \subseteq \Psi$. So we just have to prove that $m'(m''(\Delta(u_2))) \subseteq \Psi$.

 According to Lemma \ref{lem:caso3}.3, since $p\in GVar(q_i)$, either $p\in GVar(q_i)$ or there is a triple $p' \mapsto(s,p)\in \Pi_i$. In the former case, by Lemma \ref{lem:constr2}, $m'(m''(\Delta(u_2))) \subseteq Constr(Def,m'(p))$. Otherwise, if there is a triple $p' \mapsto(s,p)\in \Pi_i$, then for the same reason, 
 $m'(m''(\Delta(u_2))) \subseteq Constr(Def,m'(p')$. In any case, $m'(m''(\Delta(u_2))) \subseteq  \Psi$.

\end{enumerate}

   \end{itemize}
  
\end{proof}

 \begin{thm}[Soundness]\label{thm:sound}
If $m$ is the computed answer of $(q,\emptyset,\emptyset,\emptyset) \!\stackrel{*}{\deriv}_{\GDB,Def}\!(\emptyset,\Psi,\!\M,\!\Pi)$, then $\GDB, Def, m  \models q$. 
\end{thm}

\begin{proof}
If $q = \ell_1, \dots, \ell_j$, according to Def. \ref{def:sat}, $m: GVar(q) \to  G$ is a correct answer for $q$ if $\Psi =Cond(\ell_1,m,Def) \cup\dots \cup Cond(\ell_j,m,Def)$ is satisfiable. But according to Lemma \ref{lemma:Psi},  $\Psi$ is satisfiable, and according to Lemma \ref{lem:lem2}, $Cond(\ell_1,m,Def) \cup\dots \cup Cond(\ell_j,m,Def) \equiv \Psi $. Hence, $m$ is a correct answer for $q$.
 \end{proof}

 \begin{thm}[Completeness]\label{thm:compl}
 If $m$ is a correct answer for a query $q$, there is a successful derivation $d = (q,\emptyset,\emptyset,\emptyset) \!\stackrel{*}{\deriv}_{\GDB,Def}\!(\emptyset,\Psi,\!\M,\!\Pi)$ that has $m$ as its computed answer. %
\end{thm}

\begin{proof}

Let $m$ be a correct answer. To prove completeness, we need to show that there is a successful derivation $d$, such that $m$  is its computed answer.

We will construct $d$ inductively, selecting, at each state $\sigma_i = (q_i, \Psi_i, \M_i, \Pi_i)$, some clause in $q_i$,  using $m$ to decide which match we will use, and  we will show that, for every $i\ge 0$  in the derivation there is a state $\sigma_j$, with $i \le j$, that satisfies that for every node or edge variable $x \in Var(q_0)\cap Var(q_i)$  (a) either there is a pair 
$x \mapsto m(x) \in \M_j$ or $x \in Var(q_j)$;  (b) similarly, for every path variable $p \in GVar(q_0)\cap Var(q_i)$  either there is a triple 
$p \mapsto (m(p),\lambda)\in \Pi_j$ or $p \in GVar(q_j)$.

Notice that, by Lemma \ref {lemma:Psi} this is enough to prove the theorem. We proceed with the construction of the derivation. 

It should be obvious that $\sigma_0$ satisfies (a) and (b). 
Suppose that we are in the state $\sigma = (q, \Psi, \M, \Pi)$. By the inductive hypothesis, we may assume that $\sigma$ satisfies (a) and (b). Let $\ell$ be any clause in $q \cap q_0$. The next derivation step depends on the form of $\ell$:

\begin{itemize}
\item If $\ell= \aP\Box\Phi$, where $\aP$ is a node or edge clause, we apply the rule {\tt(R)}, 
$\sigma= (q' \cup \{\ell\}, \Psi, \M, \Pi) \deriv_{\GDB,Def}  \sigma' = (q', \Psi', \M', \Pi')$, using the match $h: Var(\aP) \to  \GDB$ defined as $h = m|_{Var(\aP)}$. This means that $\M'= \M \cup \{x \mapsto m(x) \mid x \hbox{ is a node or edge variable in } \aP\}$ and $\Pi' = \Pi$. By Def. \ref{def:rules}, this implies that $\sigma'$ satisfies (a) and (b). Notice that we can apply this rule because the assumption that $\sigma'$ is consistent is a direct consequence of the fact that $m$ is a correct answer.

\item If $\ell = (x\!:\!nt \stackrel{p:\alpha}{\pathto} x'\!:\!nt')\Box\Phi$ and  $m(p)$ has length $1$, that is, $m(p) = (m(x) , e, m(x'))$, where $e$ is an edge in $G$, whose source and target are $m(x)$ and $m(x')$, respectively.
we choose the unfolding $u = (x\!:\!nt \stackrel{y_1:et_1}{\edge} x'\!:\!nt')$ and apply rule  
 {\tt(U1)}  using the match $m|_{Var(u)}$. The rule can be applied since $m$ is a correct answer. Again, by Def. \ref{def:rules}, $\sigma'$ satisfies (a) and (b).

\item If $\ell =  (x\!:\!nt \stackrel{p:\alpha}{\pathto},x'\!:\!nt')\Box\Phi$, with $m(p) = (m(x) , e_1 e_2 \dots e_k, m(x'))$, for $k>1$.
we will apply $k$ derivation steps. The first $k-1$  steps using rule   {\tt(U2)} and the last using rule {\tt(U1)}. 
More precisely,  by induction:
\begin{itemize}
\item The first derivation step $\sigma= (q, \Psi, \M, \Pi) \deriv_{\GDB,Def}  \sigma_1 = (q_1, \Psi_1, \M_1, \Pi_1)$ selects the unfolding  $u_1= (x\!:\!nt \stackrel{y_1:et_1}{\edge} x_1 \stackrel{p_1:\alpha_1}{\pathto} x'\!:\!nt')$ 
and the match  $h_1 \!:\!Var(\aP_0)\!\to\!\GDB$, where  $\aP_0 = (x\!:\!nt \stackrel{y_1:et_1}{\edge} x_1\!)$, defined as $\{h_1(x) \mapsto m(x); h_1(x_1) \mapsto tgt(e_1), h_1(y_1) \mapsto e_1\}$. As before, the rule can be applied since $m$ is a correct answer. 
In this case, $\M' = \M \cup \{n \mapsto src(e_1); x_1 \mapsto tgt(e_1), x_1 \mapsto e_1 \}$ and $\Pi' = \Pi \cup \{p \mapsto (e_1,p_1) \}$. This means that  $\sigma_1$ satisfies  (a), but it does not yet satisfy (b).  
In addition, the rule replaces $\ell \in q$ by $\ell_1 = (x_1 \stackrel{p_1:\alpha_1}{\pathto},x'\!:\!nt')\Box \Phi$, so that the next derivation step can be applied to this clause.

Now, for the definition of the remaining $k-1$ derivation steps, we will show by induction on $k$ that in every state $\sigma_j$, $1 \le j \le k$ , we have (c) $(p \mapsto (e_1 \dots e_j), p_j) \in \Pi_j$
and (d) a clause $\ell_j =(x_{j-1}\! \stackrel{p_j:\alpha_j}{\pathto},x'\!:\!nt')\Box \Phi \in q_j$, so that we will select this clause for the next derivation step. By definition, we have that these conditions hold for $\sigma_1$.

\begin{itemize}
\item If $\ell_j $ is the selected clause, with $1< j < k$, we apply rule   {\tt(U2)} with 
unfolding  $u_{j}= (x_{j-1}\! \stackrel{y_{j}:et_j}{\edge} x_{j}\stackrel{p_{j}:\alpha_{j}}{\pathto} x'\!:\!nt')$, 
and  match  $h_{j} \!:Var(\aP'_j)\!\to\!\GDB$, where $\aP'_j =  (x_{j-1}\! \stackrel{y_{j}:et_j}{\edge} x_{j})$, defined as $\{h_j(x_{j-1}) \mapsto src(e_{j}), h_j(x_{j}) \mapsto tgt(e_{j}), h_j() \mapsto e_{j}$. Again, we can apply this rule with this match, since $m$ is a correct answer and $x_j, y_{j}$ were defined as fresh new variables in the unfolding process, so they cannot cause the inconsistency of $\sigma_{j+1}$.

 Then $\ell_{j+1} =(x_j\! \stackrel{p_{j+1}:\alpha_{j+1}}{\pathto},x'\!:\!nt')\Box \Phi$, $\M_{j+1} = \M_j \cup \{x_{j-1} \mapsto src(e_{j}); x_{j} \mapsto tgt(e_{j}), y_{j+1} \mapsto e_{j+1} \}$ and $\Pi_{j+1} = \Pi_i \cup \{p \mapsto (e_1 \dots e_{j},p_{j+1}) \}$, and $\ell_j\in q_i$ has been replaced by $\ell_{j+1}$. Therefore, $\sigma$ satisfies (a), (c), and (d).
 
 \item If $\ell_{k}$ is the selected clause, we apply rule   {\tt(U1)} with 
unfolding  $u_{k}=(x_{k-1}\!:\!nt'\stackrel{y_k:et_k}{\edge} x'\!:\!nt')$ and  match  
$h_{k} \!:\!Var(u_k)\to\!\GDB$,  defined as $\{h_k(x_{k}) \mapsto m(x_{k-1}), h_k(x') \mapsto m(x'), h_k(y_k) \mapsto e_k\}$.  This means that  $\Pi_{k}= (\Pi_{j-1} \setminus \{(p \mapsto (e_1 \dots e_k), p_{k+1})\} ) \cup \{(p \mapsto (e_1 \dots e_k), \lambda)\}$ .  Therefore, the resulting state satisfies (a) and (b).    

 \end{itemize}

\end{itemize}
\end{itemize}
\end{proof}

\begin{thm}[Correctness of computed path properties]\label{thm:corr}
 For every path variable $p:\alpha \in GVar(q)$, every property $pr \in PP$, and every successful derivation $(q,\emptyset,\emptyset,\emptyset) \!\stackrel{*}{\deriv}_{\GDB,Def}\!(\emptyset,\Psi,\!\M,\!\Pi)$ with computed answer $m$, we have $Constr(Def,m(p)) \vdash m(p).pr = v$ implies $\Psi \vdash p.pr = v$.
\end{thm}

\begin{proof}
Direct consequence of Lemma \ref{lem:lem2}
\end{proof}

We may notice that the converse is not necessarily true. The reason is that, for some $p$ and some general match, in $Constr(Def,m(p))$ there may be not enough constraints to infer that a certain property $pr$ has some value for a path $m(p)$. However, we may have $\Psi \vdash p.pr = v$, because in the set of filters $\Phi$ of some clause $\ell$ we may have the constraint $p.pr = v$.

\section{Expressiveness}
\label{sec:expr}
In this section, we study the expressiveness of our approach in comparison with other frameworks that allow specifying conditions on path elements within path queries. To the best of our knowledge, the only existing approaches offering this capability are those based on the notion of register automata \cite{KaminskiF94}. Specifically, \cite{LibkinMV16} introduces two such formalisms: Register Data Path Automata (RDPA) and Regular Expressions with Memory (REM). RDPAs are a variant of register automata, but recognizing that RDPAs may not be very user-friendly, the authors also propose REMs as a more accessible alternative. Furthermore, they demonstrate that REMs can be expressed in terms of RDPAs. Building on that work, \cite{BarceloFL15} introduces Register Logic (RL), which extends REMs.

A Register Data Path Automaton (RDPA) is a variant of finite-state automaton that includes a fixed number $k$ of registers, denoted $r_1,\dots, r_k$, two kinds of state (word states and data states)  and two transition relations (word and data transition relations) $\delta_w$ and $\delta_d$. An RDPA recognizes sequences known as data paths of the form $d_0 e_1 d_2 \dots e_{n-1} d_n$, where $e_1, \dots, e_{n-1}$ are \emph{word symbols} from a finite alphabet $\Sigma$, and $d_0, \dots, d_n$ are \emph{data symbols} from a data domain $\D$. Without loss of generality, we assume that each $d_i$ is a natural number. Additionally, we consider only data paths of length at least three (i.e., $n \geq 3$), as we focus on paths that contain at least one word symbol.

Formally, an RDPA $\mathcal{A}$ is defined as a tuple $\A = (Q,q_0, F, \tau_0, \delta)$, partitioned into data states $Q_d$ and word states $Q_w$; $q_0 \in Q_d$ is the initial state; $F\subseteq Q_w$ is the set of final states; $\tau_0 \in (\D \cup \{\bot\})^k$ is the initial register assignment, where each register is either unassigned ($\bot$) or contains a value from the data domain $\D$; and $\delta=(\delta_w, \delta_d)$ is a pair of transition relations. In particular, 
$\delta_w \subseteq Q_w \times \Sigma \times Q_d$ is the word transition relation and $\delta_d \subseteq Q_d \times \C_k \times 2^{[k]} \times Q_w$ is the data transition relation where the set $\C_k$ consists of conditions of the form:

\begin{itemize}
\item $x^=_i$: the content of the register $i$ is equal to the last read data.
\item $x^{\neq}_i$: the content of the register $i$ is different from the last read data.
\item $ v^=_i$: the constant value v is equal to the last read data.
\item $v^{\neq}_i$: the constant value v is different from the 
 last read data. 
\item $c_1 \vee c_2$, $c_1 \wedge c_2$, or $\neg c$, with the standard meaning.
\end{itemize}

Roughly speaking, to accept a data path $d_0 e_1 d_2 \dots e_{n-1} d_n$, the RDPA $\mathcal{A}$ operates as follows: $\mathcal{A}$ starts in the initial state $q_0$, with the initial configuration of registers. At any moment, if $\mathcal{A}$ is in a word state $q$ and the next symbol to \emph{read} is $e$, then $\mathcal{A}$ goes into some data state $q'$, if $(q,e,q')\in \delta_w$. But if $\mathcal{A}$ is in a data state $q$, the next symbol to read is  $d$,  $(q,c,I,q')\in \delta_d$, and  
$d$, together with the current values in the registers, satisfy condition $c$, then $\mathcal{A}$ moves to the word state $q'$, and updates the registers indexed by $I \subseteq [k]$, storing $d$ in each of those registers. Notice that this means that, at any point, the value stored in a register $r_i$ is either its initial value $\tau_0^i$ or a value $d$ read from the data path by the automaton. Finally, $\mathcal{A}$ accepts the data path if, after reading the entire sequence, the automaton reaches a final state in $F$. 

Formally, a configuration of $\A$ on a data path $w=d_0 e_1 d_2 \dots e_{n-1} d_n$ is a triple $(i,q,\tau)$, where $i$ is the current position of the automaton in $w$, indicating that the next symbol to read is $e_i$ or $d_i$), $q$ is the current state of $\mathcal{A}$, and $\tau$ is the current register assignment. The initial configuration is $(0,q_0,\tau_0)$; any configuration $(i,q,\tau)$ is a final configuration if $q \in F$; and we can move from configuration $(i,q,\tau)$ to $(i+1,q',\tau')$, written as $(i,q,\tau) \stackrel{a}{\longrightarrow} (i+1,q',\tau')$, where $a$ is the current symbol $d_i$ or $e_i$,  if one of the following holds:

\begin{itemize}
\item $q \in Q_w$, $a = e_i$ and  there is a transition $(q,e_i,q')\in \delta_w$, and $\tau' = \tau$.
\item $q \in Q_d$, $a = d_i$ and there is a transition $(q,c,I,q')\in \delta_d$, condition c is satisfied for the current data $d_i$ and the current configuration of the registers, and $\tau'$ coincides with $\tau$, except for the value of the registers $r_j$, with $j\in I$, whose value will be $d_i$.
\end{itemize}

Then $\A$ accepts the data path $w=d_0 e_1 d_2 \dots e_{n-1} d_n$, if there is a sequence of transitions $(0,q_0,\tau_0)\stackrel{\small{d_0 e_1  \dots  d_n}}{\longrightarrow} (n+1,q,\tau)$, where $q\in F$. Here, $(0,q_0,\tau_0)\stackrel{\small{d_0 e_1  \dots  d_n}}{\longrightarrow} (n+1,q,\tau)$ denotes the sequence of transition steps $(0,q_0,\tau_0)\stackrel{\small{d_0}}{\longrightarrow} (1,q_1,\tau_1)
\stackrel{\small{e_1}} {\longrightarrow} \dots \stackrel{\small{d_n}}{\longrightarrow} (n+1,q,\tau)$.

In our case, we will assume that data paths have associated properties, similar to graph paths. However, data paths are structurally simpler than graph paths, as a data path is just a sequence of symbols (word symbols and data symbols), whereas graph paths consist of a sequence of nodes and edges. As a consequence, the structure of unfoldings is also simpler for data paths.
In particular, the unfolding set $\U$ for a data path variable $p$  consists of two patterns $u_1 = (x)$ and $u_2 = (x \cdot p')$, where $x$ is the variable denoting a symbol and $p'$ is a variable denoting a data path.

In this context, a data path property definition $\dpp$ is, as in the case of graph paths, a 4-tuple $\dpp = (PP, p, \U, \Delta)$, where $PP  \subseteq \K$ is a set of path property identifiers, $p$ is the unfolded path variable, the unfolding set $\U= \{u_1,u_2\}$ and $\Delta: \U  \to 2^{Constr}$ is a function that maps each pattern $u \in \U$ to a set of constraints. 

Then, the constraints associated with a data path $w = a_0\dots a_n$, where each $a_i$ is either a word symbol or a data symbol,  $Constr(Def, w)$ is just the adaptation of Def.~\ref{def:pathconstr} to the context of data paths. That is, 
 if $n= 1$, then $Constr(Def,w) = m(\Delta(u_1))$, where $m$ is a general match of $GVar (u_1)\cup\{p\}$, defined as $\{m(x)  \mapsto a_0, m(p) \mapsto w\}$.  Otherwise, $Constr(Def,w) = m(\Delta(u_2))\cup Constr(Def,w')$
 where $w' =a_1\dots a_n$ and $m$ is a general match of $GVar (u_2)\cup\{p\}$, defined as $\{m(x) \mapsto a_0, m(p) \mapsto w, m(p') \mapsto w'$\}.   
 
Finally, we consider that specifying a language of data paths requires not only a set of path property definitions $Def$, but also a filter $\Phi$ that expresses requirements that the language must satisfy. In this sense, we consider that $Def$ and $\Phi$ recognize a sequence $w$ if $Constr(Def, w) \cup \Phi$ is consistent.

\begin{thm}
Path property definitions are more expressive than RDPA.
\end{thm}

\begin{proof}
Let us first show that there are languages that can be recognized using path properties, but that cannot be recognized using RDPAs. The key observation is that path property definitions allow expressing global constraints over entire data paths, while RDPAs are restricted to conditions based on a fixed number of registers updated incrementally during the path traversal.

Consider the language $\mathcal{L} = \{w=0~ e_1~ 0 \dots 0 ~e_{n-1}~ d_n \mid n> 1, d_n = n+1\}$, we will show that $\Le$ is not recognizable by any RDPA. That is, $\Le$ consists of the data paths where all data symbols except the last are 0, and the last data symbol equals the total number of word symbols.

Suppose that $\A = (Q,q_0, F, \tau_0, \delta)$, we will show that $\A$ cannot recognize $\Le$. 

Let $Q_p \subseteq Q_d$ be the set of states $q$ such that there is a final state $q'\in F$ and a data transition $(q,c,I,q') \in \delta_d$, for some condition c and some register update $I$. Notice that for any prefix of the form $0~ e_1~ 0 \dots 0 ~e_{n-1}$,  there must exist a derivation  $(0,q_0,\tau_0)\stackrel{1 e_1 1 \dots  1 e_n}{\longrightarrow} (n+1,q,\tau)$ with $q\in Q_p$. Otherwise, the automaton would not recognize $0~ e_1~ 0 \dots 0 ~e_{n-1}~ d$, with $d = n+1$. 

Let us translate automata conditions into ``standard'' first-order formulas over a variable $z$ representing the last read data:
\begin{itemize}

\item $trans_z(x^=_i) \equiv x_i == z$
\item $trans_z(x^{\neq}_i) \equiv x_i \neq z$
\item $ trans_z(v^=) \equiv v == z$
\item $trans_z(v^{\neq}) \equiv v\neq z$
\item $trans_z(c_1 \vee c_2) \equiv (trans_z(c_1)\vee  trans_z(c_2))$, $trans_z(c_1 \wedge c_2) \equiv (trans_z(c_1)\wedge  trans_z(c_2))$, and $trans_z(\neg c) \equiv \neg trans_z(c)$
\end{itemize}

Finally, let
$$\Psi = \bigvee_{(q,c,I,q') \in \delta_d  \wedge q' \in F} trans_z(c). $$
In other words, $\Psi$ is the disjunction of all formulas corresponding to the conditions labeling the data transitions from any state in $Q_p$  to a final state. Obviously, if $\A$ recognizes some data sequence, then $\Psi$ must hold for the values of the registers at the nth configuration and the last data symbol.

We know that $\Psi$ is equivalent to a formula $\Psi'$ in disjunctive normal form, $\Psi '= \bigvee_{1 \le i \le m} cnj_i$, for some $m$, where each $cnj_i$ is a conjunction of atomic equalities and disequalities of the form $z = v, z \neq v, z = x_j$ or $z  \neq x_j$, where $v$ is any constant and $x_j$ denotes the content of the register $r_j$, for any $j$. 

Let $max$ be the maximum of the set that includes all the initial values of $\tau^1_0, \dots, \tau^k_0$ and all the constants $v$ included in any equality or disequality in $\Psi'$. Notice that if we match $m(z) = v$, for any $v>max$, then no equality in $m(\Psi')$ involving $z$ will hold, but all disequalities also involving $z$ will.

As a consequence, if all conjunctions in $\Psi'$ include an equality on z, the automaton will not recognize any data path of the form $0~ e_1~ 0 \dots 0 ~e_{n}~ (n+1)$, when $n > max$. Conversely, if there is a conjunction $cnj$ in $\Psi'$ consisting only of equalities not involving $z$ and disequalities, then if the automaton recognizes a data path of the form $0~ e_1~ 0 \dots 0 ~e_{n}~ (n+1)$, with $n > max$, then this would mean that all equalities and disequalities in $cnj$ not involving $z$ would hold. But this would mean that $\A$ will also recognize a data path of the form $0~ e_1~ 0 \dots 0 ~e_{n} (n+2) \notin \Le$. So $\A$ would not recognize $\Le$.

Let us now see that $\Le$ can be accepted for a set of definitions of path properties $Def$, and a filter $\Phi$. 
Let $Def$ be:
\begin{itemize}
\item $PP$ consists of the  properties $length, last$
\item The unfolded path variable is $p$.
\item The unfolding set $\U$ consists of patterns $u_1= (x)$ and 
$u_2 =(x \cdot p')$ that unfold $p$. 
\item $\Delta$ is defined as follows:
\begin{itemize}
\item $\Delta(u1) = \{p.length == 1, p.last == x\}$
\item $\Delta(u2) = \{p.length ==1+p'.length, p.last == p'.last, p'.length > 1\}$
\end{itemize}
\end{itemize}

It should be clear that, given the data path $w =d_0 e_1 d_2 \dots e_{n-1} d_n$, $Constr(Def, w)$ is consistent and, moreover, $Constr(Def,w) \vdash w.last = d_{j+1}$ and $Constr(Def,w) \vdash w.length =n+1$. So, if we define the filter $\Phi = \{w.last ==  w.length\}$, we have that $Def$ and $\Phi$ recognize $\Le$.

Let us now show that every language recognized by an RDPA $\A$ can also be recognized using a set of path property definitions together with a filter $\Phi$. What we will do is define $Def$, so that it will simulate the behavior of $\A$, and define $\Phi$ so that it gives initial values to the properties of any data path $w$.

Suppose that $\A = (Q,q_0, F, \tau_0, \delta)$ is a register automaton with $k$ registers and suppose without loss of generality that states are called $q_0, \dots, q_{k'}$ and final states are $q_i$, with $k''\le  i \le k'$, for some $k''$. Let $Def$ be defined as follows:

\begin{itemize}
\item $PP$ consists of the properties $state, r_1, \dots,r_k$, where $state$ takes values in the interval $[0:k']$. 
The idea is that $w.state$ and $w.r_1, \dots,w.r_k$ will denote, respectively,  the state of $\A$ and the contents of the registers when processing $w$.
\item The unfolded path variable is $p$.
 \item As before, the unfolding set $\U$ for the pattern $(p)$ consists of patterns  $u_1= (x)$ and 
$u_2 =(x \cdot p')$.
\item Before defining $\Delta$ and $\Phi$, let us define a new translation $\Trans$, of automata conditions over free variables $z$ and $p$, which at any moment represent the recently read symbol and a given data path, respectively. Then, in this translation $p.r_i$ represents the value of the i-th register  for that data path. 

\begin{itemize}

\item For each $i:1\le i \le k$, $\Trans_{z,p}(x^=_i) \equiv p.r_i == z$
\item For each $i:1\le i \le k$,$\Trans_{z,p}(x^{\neq}_i) \equiv p.r_i \neq z$
\item $ \Trans_{z,p}(v^=) \equiv v == z$
\item $\Trans_{z,p}(v^{\neq}) \equiv v \neq z$
\item $\Trans_{z,p}(c_1 \vee c_2) \equiv (\Trans_{z,p}(c_1)\vee  \Trans_{z,p}(c_2))$, $\Trans_{z,p}(c_1 \wedge c_2) \equiv (\Trans_{z,p}(c_1)\wedge  \Trans_{z,p}(c_2))$, and $\Trans_{z,p}(\neg c) \equiv \neg \Trans_{z,p}(c)$
\end{itemize}
In addition, let us define the conjunctions of constraints that we associate with each possible transition. This conjunctions are the basis for defining $\Delta$. More precisely, a transition $t\in \delta$
 may have two kinds of associated conjunctions: $Cns(t)$ will be the conjunction of constraints associated with $t$, when $t$ is used in an intermediate state, that is, not to reach a final state; while $CnsF(t)$ will be the conjunction of constraints associated with $t$ when used in the final step. 

\begin{itemize}

\item If $(q_i,e,q_{i'})\in \delta_w$ then $$Cns(q_i,e,q_{i'})=  (p.state == i \wedge p'.state == i' \wedge e == x\wedge  (\bigwedge_{1 \le j \le k} p'.r_j == p.r_j ))$$
i.e., if $(q_i,e,q_{i'})\in \delta_w$, the constraint to simulate this transition is a conjunction stating that we are at state $i$, the recently read symbol is $e$, the new state is $i'$ and the values of the registers remain the same. 

\item If $(q_i,c,I,q_{i'})\in \delta_d$ then 

\begin{align}
Cns(q_i,c,I,q_{i'}) = &(p.state== i   \wedge \Trans_{x,p}(c) \wedge p'.state == i' \wedge \nonumber\\ 
&\wedge (\bigwedge_{(1 \le j \le k \wedge j \in I)} p'.r_j == x) \wedge (\bigwedge_{(1 \le j \le k\wedge j\notin I)} p'.r_j == p.r_j ) )\nonumber
\end{align}

i.e., if $(q_i,c,I,q_{i'})\in \delta_d$, the constraint to simulate this transition  is a conjunction stating that we are at state $i$, the translation of condition $c$ holds for the recently read symbol\footnote{When applying this definition with a given match $m$, we will have $m(x) = d$, where $d$ is the recently read symbol.} $d$ and the current values of the registers, the new state is $i'$, the values of the registers whose index is in $I$ is $d$ for the new state, and the values of the rest of the registers remain unchanged.

Similarly, we define $CnsF((q_i,c,I,q_{i'})$, for each $(q_i,c,I,q_{i'})\in \delta_d$, when this is the final transition to recognize a data path. In particular,

$$CnsF((q_i,c,I,q_{i'}) = (p.state== i \wedge \Trans_{x,p}(c) \wedge i' \ge k' \wedge i' \le k'')$$

that is, the constraints  to simulate this transition  is a conjunction stating that we are at state $i$, the translation of condition $c$ holds for the recently read symbol  and the current values of the registers, and  the resulting state is a final state. 




Then, the set of all constraints associated to all intermediate transitions in the automaton is:
$$Set(\delta) = \{(Cns(q_i,c,I,q_{i'})\mid (q_i,c,I,q_{i'})\in \delta_d\}\cup \{Cns(q_i,e,q_{i'})\mid (q_i,e,q_{i'})\in \delta_w\}$$

and the set of all constraints associated to all final transitions is:

$$SetF(\delta) = \{(CnsF(q_i,c,I,q_{i'}))\mid (q_i,c,I,q_{i'})\in \delta_d\}$$


\end{itemize}

\item Now, $\Delta$ is defined as follows:
\begin{itemize}
\item $\Delta(u_1) = \bigvee_{cnj \in SetF(t)} cnj$. The unfolding $u_1$ is used when the given data sequence consists only of one data symbol, that is, when the recognition process is expected to end. Hence, the definition of $\Delta$ just says that for one of these transitions $t$, its associated conjunction of constraints must be consistent, which implies that we have reached a final state.

\item $\Delta(u_2)= \bigvee_{ cnj \in Set(t)} cnj$. The unfolding $u_2$ is used when the given data sequence consists of more than one symbol, so the transition corresponds to an intermediate step. Hence $\Delta$ simply states that the conjunction associated with each of these transitions is a possible step.  

\end{itemize}

\item  Finally, to give initial values to the properties of a data path $w$, we define $$\Phi = \{w.state == 0, w.r_1 == \tau^1_{0},  \dots, w.r_k == \tau^k_{0}\}$$
 
 \end{itemize}

 Let us now show that the language recognized by $\A$ and the language defined by $Def$ and $\Phi$ coincide. Let $ w =a_0 \dots a_n$ be a data path. According to Lemma \ref{lem:constr1} and Def. \ref{def:pathconstr},
 $$Constr(Def,w) = \bigcup_{0\le i < n} m_i(\Delta(u_2)) \cup Constr(Def,w_n) = \bigcup_{0\le i < n} m_i(\Delta(u_2)) \cup m_{n}(\Delta(u_1))$$
where, for $0\le i < n$,  $m_i$  is the general match of $GVar(u_2)\cup \{p\}$, defined  $\{m_i(x) \mapsto a_{i}, m_i(p) \mapsto w_i, m_i(p') \mapsto w_{i+1}\}$;  $m_{n}$  is the general match of $GVar(u_1)\cup \{p\}$, defined  $\{m_{n}(x) \mapsto a_{n}, m_{n}(p) \mapsto w_n\}$; and where,  for $0\le i \le n$, $w_i =a_i \dots a_n$.  Notice that this means that $w_0 = w$.

Let $Sconstr(Def,w,i)$ be a set of consistent conjunctions of constraints in $Constr(Def,w)\cup \Phi$ of length $i+1$. Intuitively,   each conjunction $cnj \in Sconstr(Def,w,i)$ represents a derivation of the automaton of length $i$. These sets are defined inductively as follows:

\begin{itemize}

\item $Sconstr(Def,w,0) = \{\Phi\wedge m_0(cnj)\mid \Phi \wedge m_0(cnj) \hbox{  is consistent}, cnj \in Set(\delta)\}$. 

\item For each $i:1\le i < n$, $Sconstr(Def,w,i) =$  $$\{s\wedge m_i(cnj)\mid s \wedge m_i(cnj) \hbox{  is consistent}, s\in Sconstr(Def,w,i-1), cnj \in Set(\delta)\}.$$

\item $Sconstr(Def,w,n) =$  $$\{s\wedge m_{n}(cnj)\mid s \wedge m_{n}(cnj) \hbox{   is consistent}, s\in Sconstr(Def,w,n-1), cnj \in SetF(\delta)\}.$$
\end{itemize}
Notice that, given the kind of constraints in the sets $Sconstr(Def,w,i)$, checking the consistency of these conjunctions is quite simple. In particular: 
\begin{itemize}

\item If $cnj \in Sconstr(Def,w,0)$, $cnj$ is of the form 
$$w.state == 0, w.r_1 == \tau^1_{0},  \dots, w.r_k == \tau^k_{0} \wedge m_0(cnj')$$ 
where $cnj' \in Set(\delta)$
 is of the form:
\begin{align}
(p.state== i_1  \wedge \Trans_{x,p}(c) \wedge p'.state == i_2 \wedge (\bigwedge_{(1 \le j \le k \wedge j \in I)} p'.r_j == x) \wedge (\bigwedge_{(1 \le j \le k\wedge j\notin I)} p'.r_j == p.r_j ) ). \nonumber
\end{align}

So, 
\begin{align}
m_0(cnj')  = 
&(w.state== i_1  \wedge \Trans_{a_0,w}(c) \wedge w_1.state == i_2 \wedge (\bigwedge_{(1 \le j \le k \wedge j \in I)} w_1.r_j == x) \wedge \nonumber\\ &\wedge (\bigwedge_{(1 \le j \le k\wedge j\notin I)} w_1.r_j == w.r_j ) ).\nonumber
\end{align}

Then, $cnj$ is consistent iff $i_1$ is equal to $0$ and the translation of the given  condition $c$, when $p$ is $w$ and $x$ is $a_0$, $\Trans_{a_0,w}(c)$ holds. Notice that this is equivalent to say that $cnj$ must be associated with a transition $(q_0,c,I,q_{i''})\in \delta_d$, such that $c$ holds for the current configuration.

\item  If $cnj \in Sconstr(Def,w,i)$, for $0<i<n$, $cnj$ is of the form $s\wedge m_i(cnj')$, where  $s\in Sconstr(Def,w,i-1)$ and $cnj' \in Set(\delta$).  Moreover, $s$ is of the form $s'\wedge m_{i-1}(cnj'')$, where 
$cnj'' \in Set(\delta)$ and $s' = \Phi$, if $i=1$, or $s'\in Sconstr(Def,w,i-1)$, otherwise. Given the kind of constraints in these conjunctions, and assuming that  $s$ is consistent, we just have to check the consistency of 
$m_{i-1}(cnj'') \wedge m_i(cnj')$. We have two cases:

\begin{enumerate}

\item $cnj'' = (p.state == i_1 \wedge p'.state == i_2 \wedge e == x\wedge  (\bigwedge_{1 \le j \le k} p'.r_j == p.r_j ))$ so 
\begin{align}
m_{i-1}(cnj'')=&(w_{i-1}.state == i_1 \wedge w_i.state == i_2 \wedge \wedge e == a_{i-1}\wedge  
\nonumber \\  &\wedge (\bigwedge_{1 \le j \le k} w_i.r_j == w_{i-1}.r_j ))  \nonumber
\end{align}
and 
\begin{align}
cnj' = &(p.state== i'_1   \wedge \Trans_{x,p}(c) \wedge p'.state == i'_2 \wedge (\bigwedge_{(1 \le j \le k \wedge j \in I)} p'.r_j == x) \wedge \nonumber \\  &\wedge (\bigwedge_{(1 \le j \le k\wedge j\notin I)} p'.r_j == p.r_j ) ) \nonumber
\end{align}
thus 
\begin{align}
m_{i}(cnj') = &
(w_i.state== i'_1   \wedge \Trans_{a_i,w_i}(c) \wedge w_{i+1}.state == i'_2 \wedge \nonumber \\  &\wedge (\bigwedge_{(1 \le j \le k \wedge j \in I)} w_{i+1}.r_j == a_i) \wedge (\bigwedge_{(1 \le j \le k\wedge j\notin I)} w_{i+1}r_j == w_i.r_j ) ) \nonumber
\end{align}

Then, in this case, $cnj$ is consistent iff $i_2 = i'_1$ and  $\Trans_{a_i,w_i}(c)$ holds.

\item $cnj'' = $
\begin{align}  
&(p.state== i_1   \wedge \Trans_{x,p}(c) \wedge p'.state == i_2 \wedge (\bigwedge_{(1 \le j \le k \wedge j \in I)} p'.r_j == x) \wedge  \nonumber \\  &\wedge(\bigwedge_{(1 \le j \le k\wedge j\notin I)} p'.r_j == p.r_j ) ) \nonumber
\end{align}

so 
\begin{align}
m_{i-1}(cnj'') =  
&(w_{i-1}.state== i_1   \wedge \Trans_{a_{i-1},w_{i-1}}(c) \wedge w_i.state == i_2 \wedge \nonumber \\  &\wedge  (\bigwedge_{(1 \le j \le k \wedge j \in I)} w_i.r_j == x) \wedge (\bigwedge_{(1 \le j \le k\wedge j\notin I)} w_i.r_j == w_{i-1}.r_j ) )\nonumber
\end{align}

and $$cnj' =
(w_{i-1}.state == i'_1 \wedge w_i.state == i'_2 \wedge e == a_{i-1}\wedge  (\bigwedge_{1 \le j \le k} w_i.r_j == w_{i-1}.r_j ))$$

thus $$m_{i}(cnj') = (w_{i}.state == i'_1 \wedge w_{i+1}.state == i_2 \wedge e == a_{i}\wedge  (\bigwedge_{1 \le j \le k} w_{i+1}.r_j == w_{i}.r_j ))$$

Then, in this case, $cnj$ is consistent iff $i_2 = i'_1$.

\end{enumerate}

\item Finally,  if $cnj \in Sconstr(Def,w,n)$, $cnj$ is of the form $s\wedge m_n(cnj')$, where  $s\in Sconstr(Def,w,n-1)$ and $cnj' \in SetF(\delta$.  Moreover, $s$ is of the form $s'\wedge m_{n-1}(cnj'')$, where 
$cnj'' \in Set(\delta)$ and  $s'\in Sconstr(Def,w,n-1)$, otherwise. As in the previous case,  we just have to check the consistency of 
$m_{i-1}(cnj'') \wedge m_i(cnj')$. In this case, 
$$cnj'' = (p.state == i_1 \wedge p'.state == i_2 \wedge e == x\wedge  (\bigwedge_{1 \le j \le k} p'.r_j == p.r_j ))$$
so $$m_{n-1}(cnj'') = (w_{m-1}.state == i_1 \wedge w_n.state == i_2 \wedge e == a_{i-1}\wedge  (\bigwedge_{1 \le j \le k} w_n.r_j == w_{n-1}.r_j ))$$

and $$cnj' =(p.state== i'_1   \wedge \Trans_{x,p}(c) \wedge i'_2 \ge k' \wedge i' _2 \le k'')$$

thus $$m_{i}(cnj') = (w_n.state== i'_1   \wedge \Trans_{a_n,w_n}(c) \wedge i'_2 \ge k' \wedge i' _2 \le k'')$$

Hence, in this case, $cnj$ is consistent iff $i_2 = i'_1$ and  $\Trans_{a_n,w_n}(c)$ holds.

\end{itemize}

It is easy to see that, if we flatten $Sconstr(Def,w,n)$ into a formula, the result is equivalent to 
$(Constr(Def,w)\cup \Phi)$, i.e.
$$Constr(Def,w)\cup \Phi \equiv \bigvee_{cnj\in Sconstr(Def,w,n)} cnj$$
which means that $Constr(Def,w) \cup \Phi$ is consistent iff there is a conjunction $s \in Sconstr(Def,w,n)$ such that $(\bigwedge_{cnj\in s} ~ cnj)$ is consistent. 

What we will prove is that for every $i: (0 < i < n)$, $(0,q_0,\tau_0) \stackrel{a_0 \dots  a_{i-1}}{\longrightarrow} (i,q_{i_1},\tau_i)$ if and only if there is a consistent conjunction $cnj \in  Sconstr(Def,w,i-1)$, such that $cnj \vdash w_i.state = i_1$ and for every $j:1 \le j \le k$, we have $cnj \vdash w_i.r_{j} = \tau^{j}_i$. In this case, we will say that $(i,q_{i_1},\tau_{i})$ and $cnj$ are \emph{directly related}. Moreover, we will also prove that $(0,q_0,\tau_0) \stackrel{a_0 \dots  a_n}{\longrightarrow} (n+1,q,\tau)$ with $q \in F$ if and only if there is a consistent conjunction $cnj \in  Sconstr(Def,w,n)$. Notice that this completes the proof of our theorem, since it would mean that a data sequence $w$ is recognizable by an automaton $\A$, i.e.,
$(0,q_0,\tau_0) \stackrel{a_0 \dots  a_n}{\longrightarrow} (n+1,q,\tau)$,  with $q \in F$, if and only if there is a consistent conjunction $cnj \in  Sconstr(Def,w,n)$, i.e.,$Constr(Def,w)\cup \Phi$ is consistent.

We proceed by induction.
\begin{itemize}
\item Let $i = 1$, $(0,q_0,\tau_0) \stackrel{a_0}{\longrightarrow} (1,q_{i_1},\tau)$, via transition $(q_0,c,I,q_{i_1})\in \delta_d$ iff there is a consistent conjunction $cnj \in Sconstr(Def,w,0)$, of the form $\Phi \wedge m_0(cnj')$, where 
\begin{align}
cnj' = &(p.state== 0  \wedge \Trans_{x,p}(c) \wedge p'.state == i_1 \wedge (\bigwedge_{(1 \le j \le k \wedge j \in I)} p'.r_j == x) \wedge \nonumber \\  &\wedge  (\bigwedge_{(1 \le j \le k\wedge j\notin I)} p'.r_j == p.r_j ) )\nonumber   
\end{align}

So, 
\begin{align}
\Phi \wedge m_0(cnj')  = &(w.state == 0, w.r_1 == \tau^1_{0},  \dots, w.r_k == \tau^k_{0} \wedge w.state== 0  \wedge \Trans_{a_0,w}(c) \wedge \nonumber \\  &\wedge w_1.state == i_1 \wedge   (\bigwedge_{(1 \le j \le k \wedge j \in I)} w_1.r_j == a_0) \wedge (\bigwedge_{(1 \le j \le k\wedge j\notin I)} w_1.r_j == w.r_j ) ) \nonumber 
\end{align}

In particular, $c$ holds for the recently read data, $a_0$, and the initial contents of the registers, $\tau_0$, which allow the automaton to make that transition, iff $ \Trans_{a_0,w}(c)$ holds, which means that $\Phi \wedge m_0(cnj')$ is consistent. Moreover the resulting state after the transition is $q_{i_1}$ and the contents of each register $r_j$  after the transition are equal to $a_0$ or to $\tau^j$, depending whether $j$ is in $I$ or not, iff $cnj \vdash w_1.state ==i_i$ and for every $j:1 \le j \le k$, we have $cnj \vdash w_1r_{j} == a_0$ or $cnj \vdash w_1r_{j} == \tau^{j}_0$, depending whether $j\in I$ or not.

\item Let $1 < i <n$. By the inductive hypothesis,  we may assume that, for every $(i-1,q_{i_2},\tau_{i-1})$, there is a consistent conjunction $cnj_1 \in Sconstr(Def,w,i-2)$, such that $(i-1,q_{i_2},\tau_{i-1})$ and $cnj_1$ are directly related, and vice versa. 

This means that it is enough to see that $(i-1,q_{i_1},\tau_{i-1}) \stackrel{a_{i-1}}{\longrightarrow} (i,q_{i_2},\tau_i)$ iff there is a consistent conjunction $cnj \in Sconstr(Def,w,i-1)$, of the form $cnj_1\wedge m_{i-1}(cnj_2)$, such that $(i,q_{i_2},\tau_i)$ and $cnj$ are directly related, assuming that $(i-1,q_{i_1},\tau_{i-1})$ and $cnj_1$ are directly related. There are two cases, depending whether the transition $ (i-1,q_{i_1},\tau_{i-1}) \stackrel{a_{i-1}}{\longrightarrow} (i,q_{i_2},\tau_i)$ is a word transition or a data transition.

\begin{enumerate}
\item The transition is defined by the  rule $(q_{i_1},a_{i-1},q_{i_2})\in \delta_w$ iff there is a conjunction $cnj_2\in Set(\delta)$, with $$cnj_2 =  (p.state == i_1 \wedge p'.state == i_2 \wedge a_{i-1} == x\wedge  (\bigwedge_{1 \le j \le k} p'.r_j == p.r_j )).$$ Thus $$m_{i-1}(cnj_2) =(w_{i-1}.state == i_1 \wedge w_{i}.state == i_2 \wedge a_{i-1} == a_{i-1}\wedge  (\bigwedge_{1 \le j \le k} w_{i}.r_j == w_{i-1}.r_j )).$$

In particular,  the resulting state after the transition is $q_{i_2}$ and the contents of each register $r_j$  after the transition are equal to $\tau^j_{i-1}$,  for every $j:1 \le j \le k$ iff we have $cnj \vdash w.r_{j} == \tau^{j}_i$

\item The transition is defined by the  rule  $(q_{i_1},c,I,q_{i_2})\in \delta_d$ iff there is a conjunction $cnj_2 \in Set(\delta)$, with 
\begin{align}
cnj_2=
&(p.state== i_1   \wedge \Trans_{x,p}(c) \wedge p'.state == i_2 \wedge  \nonumber \\  &\wedge (\bigwedge_{(1 \le j \le k \wedge j \in I)} p'.r_j == x) \wedge(\bigwedge_{(1 \le j \le k\wedge j\notin I)} p'.r_j == p.r_j ) )\nonumber. 
\end{align}

Thus 
\begin{align}
m_{i-1}(cnj_2) =
&(w_{i-1}.state== i_1   \wedge \Trans_{a_{i-1},w_{i-1}}(c) \wedge w_i.state == i_2 \wedge \nonumber \\  &\wedge (\bigwedge_{(1 \le j \le k \wedge j \in I)} w_i.r_j == a_{i-1})  \wedge (\bigwedge_{(1 \le j \le k\wedge j\notin I)} w_i.r_j == w_{i-1}.r_j ) )\nonumber  
\end{align}

In particular, $c$ holds for the recently read data, $a_0$, and the  contents of the registers, $\tau_{i-1}$, which allow the automaton to make that transition, iff $ \Trans_{a_{i-1},w_{i-1}}(c)$ holds, which means that $m_{i-1}(cnj_2)$ is consistent. Moreover the resulting state after the transition is $q_{i_2}$ and the contents of each register $r_j$  after the transition is equal to $a_{i-1}$ or to $\tau^j_{i-1}$, depending whether $j$ is in $I$ or not, iff $cnj' \vdash w_i.state =i_2$ and for every $j:1 \le j \le k$, we have $cnj \vdash w_i.r_{j} == a_{i-1}$ or $cnj \vdash w_ir_{j} == \tau^{j}_0$, depending whether $\in I$ or not.

\end{enumerate}

\end{itemize}

To finish the proof, we only have to prove $(0,q_0,\tau_0) \stackrel{a_0 \dots  a_n}{\longrightarrow} (n+1,q,\tau)$ with $q \in F$ if and only if there is a consistent conjunction $cnj \in  Sconstr(Def,w,n)$. We know that $(0,q_0,\tau_0) \stackrel{a_0 \dots  a_{n-1}}{\longrightarrow} (n,q_{i_1},\tau_n)$ if and only if there is a consistent conjunction $cnj \in  Sconstr(Def,w,n-1)$.

We have $(n,q_{i_1},\tau_n) \stackrel{a_n}{\longrightarrow} (n+1,q_{i_2},\tau)$, via transition $(q_{i_1},c,I,q_{i_2})\in \delta_d$ iff $cnj_2 \in SetF(\delta)$, with $$cnj_2=
(p.state== i_1   \wedge \Trans_{x,p}(c) \wedge  i_2 \ge k' \wedge i_2 \le k''))$$
Thus $$m_{n}(cnj_2) =(w_{n}.state == i_1 \wedge \Trans_{a_n,w_n}(c) \wedge  i_2 \ge k' \wedge i_2 \le k'')$$
but the transition to a final state is possible iff $c$ holds for the recently read data, $a_n$, and the contents of the registers, $\tau_n$, which allow the automaton to make that transition and if $q_{i_2}\in F$. But this is equivalent to saying that $ \Trans_{a_0,w}(c)$ is true, and that $cnj \vdash i_2 \ge k' \wedge i_2 \le k''$, that is, $cnj$ is consistent. 
\end{proof}

Our path properties are also more expressive than REMs \cite{LibkinMV16}, since REMs can be  translated into RDPAs. However, RL \cite{BarceloFL15} and our path properties are incomparable, since a similar counter-example would serve to prove that there are queries defined by our path properties, that are not expressible by RL but, on the contrary, the query: \emph{Find pairs of nodes $x$ and $y$, such that there is a node $z$ and a path $p$ from $x$ to $y$ in which each node is connected to $z$} is not expressible in our formalism, but is expressible in RL. Unfortunately, implementing RL would not be feasible in practice, since its combined complexity is EXPSPACE-complete and its data complexity is in PSPACE \cite{BarceloFL15}.

\section{Complexity and empirical analysis}
\label{sec:comp}
In this section, we analyze our proposal from two perspectives. First, we discuss the complexity of our approach. Then, we review the results obtained by executing simple path queries on our running example using a prototype implementation in Prolog. These results suggest that our approach is valuable not only for its added expressiveness but also because it can significantly enhance the performance of path query execution.

\subsection{Complexity}
The first problem in analyzing the complexity of our approach lies in the fact that, in general, determining whether a set of arbitrary constraints is satisfiable is undecidable. A natural solution is to restrict the class of constraints to those for which satisfiability is not only decidable but also computationally feasible. For instance, solving systems of linear equations over the reals is decidable and can be done using Gaussian elimination with quadratic complexity \cite{JaffarM94}. However, if we allow equality constraints over values from a finite domain, the satisfiability problem is almost always NP-hard. A similar situation arises with word equations, where satisfiability is also NP-hard, although it is not even known whether the problem is in NP \cite{JaffarM94}.

However, for example, relying exclusively on linear constraints over the reals is not always practical, as has been observed in the context of Constraint Logic Programming (CLP). For example, consider a constraint like
$x^2 ==y$, which is non-linear and, as a consequence, falls outside the allowed class of linear constraints. Yet, in CLP, this is not treated as a major obstacle. The key reason is that, rather than attempting to fully solve the set of constraints, or check their satisfiability, at every step of the computation, CLP systems typically simplify constraints incrementally until reaching a \emph{solved form}. 
For instance, if the constraint set includes $x^2 ==y$, the simplification procedure would initially leave it unsolved, assuming it is satisfiable. If, later in the computation, a constraint like $x==5$ is added, the simplification process would replace $x^2 ==y$ with $y == 25$, potentially allowing the system to now fully solve the constraints or detect inconsistencies. We refer to this approach as \emph{partial constraint solving}.

If we abstract away from the problem of solving constraint sets or checking their satisfiability, treating these tasks as if handled by an oracle, we can more easily analyze the complexity involved in collecting the relevant set of constraints required to compute the properties of paths.

More precisely, according to our operational semantics, when solving a path query of the form $(x\!:\!nt \stackrel{p:\alpha}{\pathto},x'\!:\!nt')\Box\Phi$ with respect to a path property definition $Def$, at each step  we add to the current set of constraints $\Psi_i$ the constraints associated with the current unfolding, starting from $\Psi_0 = \emptyset$. 
If $m$ denotes the number of constraints defined for each unfolding in $Def$, and if the solution for $p$ corresponds to a path of length $n$, then the overall space used in this computation would be $(m + n*m)*n/2$. However, $(m + n*m)*n/2$ is independent of the size of the database, which implies that, disregarding the complexity of constraint solving, the data complexity of our technique remains the same as that of CRPQ queries, which is known to be  NLOGSPACE-complete \cite{consens1990graphlog,FigueiraL15}.

 Similarly, according to \cite{consens1990graphlog,FigueiraL15} the combined complexity of CRPQ queries is NP-Complete. As discussed, the cost of gathering the constraints required to compute the properties of a path query $(x\!:\!nt \stackrel{p:\alpha}{\pathto},x'\!:\!nt')\Box\Phi$ is $(m + n*m)*n/2$, where $m$ is the number of constraints defined for each unfolding in $Def$ and $n$ is the length of the path found for $p$. In this case, we can consider that $Def$ is part of the query, which means that this additional cost is polynomic in the size of the query. Therefore, incorporating the gathering of constraints for computing path properties does not increase the combined complexity of CRPQs, which remains NP-complete.
 
In conclusion, providing a detailed complexity analysis of our approach is not straightforward, due to its reliance on the chosen constraint domain and the use of partial constraint solving. For this reason, a thorough complexity study is left as future work.

\subsection{Implementation and empirical analysis}
To implement our ideas following our operational semantics, a search procedure is required to compute all query answers by exploring all possible matchings of the involved clauses.
Furthermore, this search procedure must employ a constraint solver to verify the consistency of states after simulating the application of derivation rules.
 However, since consistency checking can be computationally expensive, we can adopt the strategy used in Constraint Logic Programming (CLP) languages, as previously mentioned, where, instead of checking consistency at each step, the system simplifies the current set of constraints until it reaches a solved form.
 
 Usually, inconsistencies become apparent once constraints are in solved form. Importantly, if an inconsistency goes undetected during simplification, it does not compromise the soundness or completeness of the implementation. 

Following these ideas, we have developed a prototype implementation of the graph database used as a running example in this paper, using Ciao Prolog \cite{Ciao12}. With this implementation, we compared the performance of several basic queries executed with and without the use of different path properties. Our goal was to evaluate whether the additional overhead introduced by handling and solving constraints would make their use impractical. However, as shown below, the results demonstrate that incorporating path properties not only remains practical but also leads to much better performance compared to executing queries without them.

Specifically, in our database examples, we included only nodes of type Airport and edges of type Flight. We created four database instances: GDB2, GDB5, GDB10, and GDB50. Each instance contains 100 nodes but differs in the number of edges—200, 500, 1,000, and 5,000 edges, respectively, with $GDB2 \subseteq GDB5 \subseteq GDB10 \subseteq GDB50$. The properties of nodes and edges were randomly generated.

We executed several variations of 10 queries on each of the four database instances. The experiments were run on a MacBook Air equipped with an M1 processor and 16 GB of memory. All queries followed the form:
$$(x_1:\!Airport\stackrel{p:Flight^+}{\pathto}x_2:Airport)  \Box \  x_1.loc == A_1,  x_2.loc ==A_2$$
where, for each of the 10 queries, the departure and arrival airports ($A_1$ and $A_2$) were randomly chosen and distinct from each other.

In addition, for each of these 10 basic queries, we considered 8 different variations based on the additional constraints used.  In particular:
\begin{enumerate} 
\item No path properties and no additional constraints.
\item Path properties $length$, $cost$ and $start$, are computed. The additional filter applied is $p.length < 3$.
\item Same as (2), but the additional filter is $p.length < 5$
\item Same as (2), but the additional filter is $p.length < 10$
\item Same as (2), but two additional filters are applied, $p.length < 3, p.cost <10000$
\item Same as (5), but the additional filters are $p.length < 5, p.cost <10000$
\item Same as (5), but the additional filters are  $p.length < 10, p.cost <10000$
\item The constraint $p'.start - y.arr > 120$ (ensuring the \emph{connection time} of more than two hours between consecutive flights) is included in the definition of path properties (not in the rest of the cases), but no additional filter is added.
\end{enumerate} 

 Table \ref{tab:tablaMedias} summarizes the results, showing the average execution time (in seconds) and the average number of results for each query type. A star (*) indicates a time-out, which occurred if a query exceeded the two-hour limit. Detailed results for all individual queries can be found in the appendix.

\begin{table}[h!]
\centering
\begin{tabular}{ | c | c | c | c | c | c | c | c | c |  }
\hline
  & \#Results &Time & \#Results &Time& \#Results &Time& \#Results &Time\\
& GDB50  &GDB50& GDB10  &GDB10& GDB5  &GDB5& GDB2&GDB2  \\
\hline
No constraints. & * & * & * &  * &  * &  * &7582 & 19.318 \\
$L < 3$& 22 &0.083& 2 & 0.007& 0 & 0.002& 0 & 0.001 \\
$L < 5$& 53974&153.687 & 89 &  0.292 & 9 & 0.040 & 1& 0.003 \\
$L < 10$ & * & * & * &  * &  26706 & 71.770& 3  &0.045 \\
$L < 3, C<10000$& 22 &0.061& 2 & 0.005& 1 & 0.002& 0 & 0.001 \\
$L < 5, C<10000$& 31749&61.908 & 52 &  0.145 & 7 & 0.021 & 1& 0.002 \\
$L < 10, C<10000$ & * & * & 16199 &  74.818 &  55 & 0.145& 6  &0.037 \\
$Dep-Arr\ge 120$ & 27& 2.468 & 1 &  0.021 &  0 &  0.004 &0 & 0.001 \\
\hline
\end{tabular}
\caption{Results on average}
\label{tab:tablaMedias}
\end{table}

More specifically, each row in Table ~\ref{tab:tablaMedias} corresponds to a particular type of query. The first row shows results for queries without path property computations or additional filters. The next three rows report results when only a length limit is applied, restricting paths to at most 2, 4, or 9 edges, respectively. The following three rows add a cost constraint, limiting the path cost to less than 10,000, in addition to the length restriction. Finally, the last row introduces the connection time constraint $p'.start - y.arr > 120$. In all cases, we ensure (and verify) that the paths are simple, meaning no nodes are repeated.

Additionally, the columns in the table correspond to the results obtained for each of the four database instances considered. Specifically, for each instance, we report both the number of paths found and the total time required to compute all these results.

Examining the specific data in the table, we observe a substantial difference in execution times between queries without path properties and those using them. Without path properties, results were obtained only for the smallest database instance; in contrast, when path properties or constraints were used, results were obtained for all database instances in most cases, with only a few exceptions.

We also observe that applying the constraint that filters out ill-formed paths, i.e., those not satisfying $p'.start - y.arr > 120$, such as paths where a flight departs before the previous one arrives, is highly effective. 
While this may not hold as strongly in other types of databases where ill-formed paths may be less frequent, it illustrates the practical value of characterizing undesired paths using explicit constraints.

Another point we would like to highlight is that, for paths of at most two edges, the performance results with and without the additional cost filter are quite similar. However, this difference becomes more noticeable as path length increases. Clearly, this is a consequence of the correlation between the cost and the length of a path, implying that the cost constraint only significantly filters out paths when they are long. Importantly, for short paths where the cost constraint does not actually filter any results, the execution times remain comparable; in fact, somewhat surprisingly, including the cost constraint  yields even shorter execution times. This indicates that the overhead of computing a simple path property ($cost$ in this case) is negligible in practice.

Finally, we are not particularly surprised by these results. Without path properties that enable us to filter out a substantial number of potential answers to a given query, the number of possible solutions can grow excessively. For instance, consider the query in Example ~\ref{ex:example}, where we aim to find connections from Barcelona to Los Angeles. If the number of flights in the database is large enough, the number of valid paths can grow exponentially with the size of the database. In particular, many of these paths would involve connections that pass through nearly all the airports in the database before arriving at Los Angeles. This combinatorial explosion is precisely why we observe timeouts in all cases—except when the database is small, containing only 200 edges.

\section{Conclusion}
\label{sec:conc}

In this work, we have presented a framework for defining path properties, enabling the use of such properties to filter out many unwanted solutions and thereby substantially increasing the expressive power of query languages. Moreover, we have shown how constraints can be employed to discard solutions considered ill-formed. For simplicity, our approach has been presented as an extension of CRPQs.  

In our framework, we decouple the description of the structure of a path, expressed by regular expressions, from the specification of its properties. We believe that this makes our approach easy to use. Moreover, we believe that this decoupling will facilitate adapting the framework to support other types of navigation patterns beyond CRPQs, including more expressive query constructs. In particular, we conjecture that the key issue would be to define suitable unfolding sets to specify the properties. In this regard, we consider our proposal especially relevant at a time when the community is working towards establishing a query language standard \cite{GQL23, LibkinGPC}.

Additionally, even if we have not made a detailed study of the complexity of our approach, the empirical results obtained with our implementation show that we can considerably increase the execution performance of navigational queries when they include path patterns.  

In this sense, we think there are two aspects that would need futher research. First, a detailed study of the complexity of our approach may provide us with deeper insights about our framework. Second, it would be valuable to explore the application of these ideas within the context of GQL, which we believe may not be straightforward given the complex semantics of that query language.

\bibliographystyle{alphaurl}
\bibliography{refs}
\newpage
\begin{center}
    \sc{Appendix}
\end{center}
Below you can find the detailed results obtained when running all the queries. As explained above, all these queries were of the form $(x_1:\!Airport\stackrel{p:Flight^+}{\pathto}x_2:Airport)$, but were considered with different filters and different source and destination airports. In addition, they were executed for four databases of different size.

More precisely, each of the tables below shows all the results for a given pair of airports $x_1$ and $x_2$, with all the filters considered, and for the four database instances.

We may notice that in the last table, the results for the database instance with 200 edges are the same for all variants of the query. The reason is that in this case there are no paths between $x_1$ and $x_2$, and the execution times were all below 1 msec.

\begin{table}[h!]
\centering
\begin{tabular}{ | c | c | c | c | c | c | c | c | c |  }
\hline
  & \#Results &Time & \#Results &Time& \#Results &Time& \#Results &Time\\
& GDB50  &GDB50& GDB10  &GDB10& GDB5  &GDB5& GDB2&GDB2  \\
\hline
No constraints. & * & * & * &  * &  * &  * & 834 & 0.589 \\
$L < 3$& 31 &0.087&4& 0.007& 2 & 0.003& 1 & 0.002 \\
$L < 5$& 62549&166.442 & 99 &  0.363 & 4 & 0.053 & 1& 0.005 \\
$L < 10$ & * & * & * &  * &  12358 & 75.604& 3  &0.068 \\
$L < 3, C<10000$& 31 &0.068& 4 & 0.007& 2 & 0.003&1 & 0.002 \\
$L < 5, C<10000$& 36562&62.837 & 74 &  0.179 & 3 & 0.026 & 1& 0.005 \\
$L < 10, C<10000$ & * & * & 33780 &  98.645 &  81 & 0.509& 1  &0.011 \\
$Dep-Arr\ge 120$ & 27& 2.139 & 2 &  0.027 &  2 &  0.007 & 1 & 0.003 \\
\hline
\end{tabular}
\caption{Results query 1}
\label{tab:tabla1}
\end{table}

\begin{table}[h!]
\centering
\begin{tabular}{ | c | c | c | c | c | c | c | c | c |  }
\hline
  & \#Results &Time & \#Results &Time& \#Results &Time& \#Results &Time\\
& GDB50  &GDB50& GDB10  &GDB10& GDB5  &GDB5& GDB2&GDB2  \\
\hline
No constraints. & * & * & * &  * &  * &  * & 2219 & 0.861 \\
$L < 3$& 17 &0.084& 2 & 0.006& 1 & 0.003& 0 & 0.001 \\
$L < 5$& 55975&144.468 & 62 &  0.178 & 11 & 0.035 & 2 & 0.005 \\
$L < 10$ & * & * & * &  * &  15361 & 45.348& 8  &0.042 \\
$L < 3, C<10000$& 17 &0.066& 2 & 0.004& 1 & 0.003& 0 & 0.001 \\
$L < 5, C<10000$& 29101&54.358 & 45 &  0.095 & 10 & 0.020 & 2 & 0.004 \\
$L < 10, C<10000$ & * & * & 12194 &  34.132 &  48 & 0.186& 2  &0.003 \\
$Dep-Arr\ge 120$ & 33& 2.364 & 2 &  0.015 &  1 &  0.003 &0 & 0.001 \\
\hline
\end{tabular}
\caption{Results query 2}
\label{tab:tabla2}
\end{table}

\begin{table}[h!]
\centering
\begin{tabular}{ | c | c | c | c | c | c | c | c | c |  }
\hline
  & \#Results &Time & \#Results &Time& \#Results &Time& \#Results &Time\\
& GDB50  &GDB50& GDB10  &GDB10& GDB5  &GDB5& GDB2&GDB2  \\
\hline
No constraints. & * & * & * &  * &  * &  * & 0 & 3.626 \\
$L < 3$& 24 &0.088& 3 & 0.007& 0 & 0.003& 0 & 0.002 \\
$L < 5$& 68729&159.129 & 100 &  0.311 & 5 & 0.065 & 0& 0.006 \\
$L < 10$ & * & * & * &  * &  12832 & 88.835& 0  &0.056 \\
$L < 3, C<10000$& 24 &0.065& 3 & 0.006& 0 & 0.004&0 & 0.001 \\
$L < 5, C<10000$& 39842&60.054 & 72 &  0.152 & 3 & 0.030 & 0& 0.004 \\
$L < 10, C<10000$ & * & * & 20269 &  76.261 &  35 & 0.512& 0  &0.008 \\
$Dep-Arr\ge 120$ & 20& 3.285 & 0 &  0.016 &  0 &  0.004 &0 & 0.001 \\
\hline
\end{tabular}
\caption{Results query 3}
\label{tab:tabla3}
\end{table}

\begin{table}[h!]
\centering
\begin{tabular}{ | c | c | c | c | c | c | c | c | c |  }
\hline
  & \#Results &Time & \#Results &Time& \#Results &Time& \#Results &Time\\
& GDB50  &GDB50& GDB10  &GDB10& GDB5  &GDB5& GDB2&GDB2  \\
\hline
No constraints. & * & * & * &  * &  * &  * & 0 & 39.689\\
$L < 3$& 19 &0.093&1& 0.010& 1 & 0.005& 0 & 0.001 \\
$L < 5$& 56064&156.277 & 46 &  0.275 & 13 & 0.053 & 0& 0.004 \\
$L < 10$ & * & * & * &  * &  20924 & 118.915& 0  &0.081 \\
$L < 3, C<10000$& 19 &0.069& 1 & 0.005& 1 & 0.004&0 & 0.001 \\
$L < 5, C<10000$& 33031&60.841 & 60 &  0.179 & 11 & 0.034 & 0& 0.003 \\
$L < 10, C<10000$ & * & * & 16079 &  108.098 &  326 & 0.882& 0  &0.004 \\
$Dep-Arr\ge 120$ & 30& 2.917 & 2 &  0.028 &  1 &  0.007 &0 & 0.001 \\
\hline
\end{tabular}
\caption{Results query 4}
\label{tab:tabla4}
\end{table}

\begin{table}[h!]
\centering
\begin{tabular}{ | c | c | c | c | c | c | c | c | c |  }
\hline
  & \#Results &Time & \#Results &Time& \#Results &Time& \#Results &Time\\
& GDB50  &GDB50& GDB10  &GDB10& GDB5  &GDB5& GDB2&GDB2  \\
\hline
No constraints. & * & * & * &  * &  * &  * & 5244 & 1.045 \\
$L < 3$& 16 &0.075& 1 & 0.005& 0 & 0.002& 0 & 0.001 \\
$L < 5$& 35537&105.847 & 70 &  0.184 & 2 & 0.021 & 1& 0.000 \\
$L < 10$ & * & * & * &  * &  84405 & 33.772& 4  &0.024 \\
$L < 3, C<10000$& 16 &0.055& 1 & 0.005& 0 & 0.001& 0 & 0.001 \\
$L < 5, C<10000$& 33031&59.612 & 33 &  0.138 & 1 & 0.011 & 1& 0.001 \\
$L < 10, C<10000$ & * & * & 2272 &  20.857 &  28 & 0.141& 1  &0.001 \\
$Dep-Arr\ge 120$ & 12& 1.495 & 0 &  0.017 &  0 &  0.001 &0 & 0.001 \\
\hline
\end{tabular}
\caption{Results query 5}
\label{tab:tabla5}
\end{table}

\begin{table}[h!]
\centering
\begin{tabular}{ | c | c | c | c | c | c | c | c | c |  }
\hline
  & \#Results &Time & \#Results &Time& \#Results &Time& \#Results &Time\\
& GDB50  &GDB50& GDB10  &GDB10& GDB5  &GDB5& GDB2&GDB2  \\
\hline
No constraints. & * & * & * &  * &  * &  * & 0 & 140.483 \\
$L < 3$& 21 &0.080&1& 0.009& 0 & 0.003& 0 & 0.001 \\
$L < 5$& 57492&150.968 & 70 &  0.292 & 2 & 0.038 & 0& 0.003 \\
$L < 10$ & * & * & * &  * &  13439 & 72.103& 0  &0.087 \\
$L < 3, C<10000$& 21 &0.063& 1 & 0.007& 0 & 0.003& 0 & 0.001 \\
$L < 5, C<10000$& 36669&60.921 & 33 &  0.155 & 2 & 0.175 & 0& 0.003 \\
$L < 10, C<10000$ & * & * & 14739 &  89.451 &  44 & 0.411& 0  &0.006 \\
$Dep-Arr\ge 120$ & 32& 3.561 & 0 &  0.013 &  0 &  0.004 &0 & 0.001 \\
\hline
\end{tabular}
\caption{Results query 6}
\label{tab:tabla6}
\end{table}

\begin{table}[h!]
\centering
\begin{tabular}{ | c | c | c | c | c | c | c | c | c |  }
\hline
  & \#Results &Time & \#Results &Time& \#Results &Time& \#Results &Time\\
& GDB50  &GDB50& GDB10  &GDB10& GDB5  &GDB5& GDB2&GDB2  \\
\hline
No constraints. & * & * & * &  * &  * &  * & 730 & 1.050 \\
$L < 3$& 26 &0.087&2& 0.004& 1 & 0.002& 1 & 0.001 \\
$L < 5$& 31756&90.990 & 95 &  0.196 & 9 & 0.023 & 1& 0.001 \\
$L < 10$ & * & * & * &  * &  29622 & 118.510& 3  &0.030 \\
$L < 3, C<10000$& 26 &0.058& 2 & 0.004& 1 & 0.002&1 & 0.001 \\
$L < 5, C<10000$& 36562&62.837 & 74 &  0.179 & 3 & 0.026 & 1& 0.005 \\
$L < 10, C<10000$ & * & * & 22738 &  78.880 &  136 & 0.560& 1  &0.002 \\
$Dep-Arr\ge 120$ & 51& 2.900 & 0 &  0.029 &  0 &  0.007 &0 & 0.001 \\
\hline
\end{tabular}
\caption{Results query 7}
\label{tab:tabla7}
\end{table}

\begin{table}[h!]
\centering
\begin{tabular}{ | c | c | c | c | c | c | c | c | c |  }
\hline
  & \#Results &Time & \#Results &Time& \#Results &Time& \#Results &Time\\
& GDB50  &GDB50& GDB10  &GDB10& GDB5  &GDB5& GDB2&GDB2  \\
\hline
No constraints. & * & * & * &  * &  * &  * & 3518 & 1.002 \\
$L < 3$& 20 &0.075& 0 & 0.005& 0 & 0.001& 0 & 0.001 \\
$L < 5$& 43149&114.611 & 45 &  0.211 & 3 & 0.023 & 0& 0.004 \\
$L < 10$ & * & * & * &  * &  9010 & 25.701& 0  &0.033 \\
$L < 3, C<10000$& 20 &0.059& 0 & 0.005& 0 & 0.001&0 & 0.001 \\
$L < 5, C<10000$& 20863 &42.148 & 19 &  0.110 & 2 & 0.010 & 0& 0.002 \\
$L < 10, C<10000$ & * & * & 4946 &  40.152 &  8 & 0.238& 0  &0.005 \\
$Dep-Arr\ge 120$ & 20& 1.455 & 0 &  0.015 &  0 &  0.002 &0 & 0.001 \\
\hline
\end{tabular}
\caption{Results query 8}
\label{tab:tabla8}
\end{table}

\begin{table}[h!]
\centering
\begin{tabular}{ | c | c | c | c | c | c | c | c | c |  }
\hline
  & \#Results &Time & \#Results &Time& \#Results &Time& \#Results &Time\\
& GDB50  &GDB50& GDB10  &GDB10& GDB5  &GDB5& GDB2&GDB2  \\
\hline
No constraints. & * & * & * &  * &  * &  * & 63278 & 4.830 \\
$L < 3$& 22 &0.079& 2& 0.005& 0 & 0.002& 0 & 0.000 \\
$L < 5$& 44864&186.220 & 131 &  0.366 & 24 & 0.039 & 0& 0.001 \\
$L < 10$ & * & * & * &  * &  50081 & 100.610& 18  &0.029 \\
$L < 3, C<10000$& 22 &0.060& 2 & 0.003& 0 & 0.001&0 & 0.000 \\
$L < 5, C<10000$& 25976&76.740 & 80 &  0.170 & 18 & 0.018 & 0& 0.001 \\
$L < 10, C<10000$ & * & * & 16346 &  137.460 &  237 & 0.780& 1  &0.002 \\
$Dep-Arr\ge 120$ & 21& 2.394 & 1 &  0.025 &  0 &  0.008 &0 & 0.001 \\
\hline
\end{tabular}
\caption{Results query 9}
\label{tab:tabla9}
\end{table}

\begin{table}[h!]
\centering
\begin{tabular}{ | c | c | c | c | c | c | c | c | c |  }
\hline
  & \#Results &Time & \#Results &Time& \#Results &Time& \#Results &Time\\
& GDB50  &GDB50& GDB10  &GDB10& GDB5  &GDB5& GDB2&GDB2  \\
\hline
No constraints. & * & * & * &  * &  * &  * & 0 & 0.000 \\
$L < 3$& 26 &0.080& 1 & 0.007& 0 & 0.002& 0 & 0.000 \\
$L < 5$& 53713&125.417 & 110 &  0.223 & 5 & 0.024 & 0& 0.000 \\
$L < 10$ & * & * & * &  * &  19025 & 38.299& 0  &0.000 \\
$L < 3, C<10000$& 26 &0.060& 1 & 0.005& 0 & 0.003& 0 & 0.000 \\
$L < 5, C<10000$& 30658&50.578 & 65 &  0.120 & 4 & 0.018 & 0& 0.000 \\
$L < 10, C<10000$ & * & * & 18627 &  64.240 &  150 & 0.363& 0  &0.000 \\
$Dep-Arr\ge 120$ & 28& 2.171 & 0 &  0.027 &  0 &  0.003 &0 & 0.000 \\
\hline
\end{tabular}
\caption{Results query 10}
\label{tab:tabla10}
\end{table}

\end{document}